\documentclass[final,1p,times,authoryear]{elsarticle}
\usepackage{amsmath}
\usepackage{amsthm}
\usepackage{amssymb}
\usepackage{amsfonts}
\usepackage{dsfont}

\usepackage{float}
\usepackage[plainpages=false,pdfpagelabels=true,colorlinks=true,citecolor=blue,hypertexnames=false]{hyperref}
\usepackage[ruled,vlined]{algorithm2e}
\usepackage{mathtools}
\usepackage{extarrows}
\usepackage{url}
\usepackage{balance}
\usepackage{multirow}
\usepackage{color}
\usepackage{diagbox}
\usepackage{booktabs}
\usepackage{makecell}
\newcolumntype{V}{!{\vrule width 1pt}}

\newtheorem{theorem}{Theorem}
\newtheorem{lemma}[theorem]{Lemma}

\newtheorem{definition}[theorem]{Definition}

\newtheorem{example}[theorem]{Example}
\newtheorem{remark}[theorem]{Remark}

\newcommand{\gr}{Gr\"{o}bner }

\newcommand{\vh}{\vspace*{4pt}}

\def\gr{Gr\"{o}bner }
\def\V{{\mathbb{V}}}

\def\k{L}

\begin{document}

\begin{frontmatter}

\title{Two Algorithms for Computing Rational Univariate Representations of Zero-Dimensional Ideals with Parameters}

\author[klmm,ucas]{Dingkang Wang}
\ead{dwang@mmrc.iss.ac.cn}

\author[etown]{Jingjing Wei}
\ead{weijingjing18@mails.ucas.ac.cn}

\author[hnu]{Fanghui Xiao}
\ead{xiaofanghui@hunnu.edu.cn}

\author[klmm,ucas]{Xiaopeng Zheng}
\ead{zhengxiaopeng@amss.ac.cn}

\address[klmm]{KLMM, Academy of Mathematics and Systems Science, Chinese Academy of Sciences, Beijing 100190, China}

\address[ucas]{School of Mathematical Sciences, University of Chinese Academy of Sciences, Beijing 100049, China}

\address[etown]{Beijing ETown Academy, Beijing 100176, China}

\address[hnu]{MOE-LCSM, School of Mathematics and Statistics, Hunan Normal University, Hunan 410081, China}

\begin{abstract}
 Based on the partition of parameter space, two algorithms for computing the rational univariate representation of zero-dimensional ideals with parameters are presented in the paper. Unlike the rational univariate representation of zero-dimensional ideals without parameters, the number of zeros of zero-dimensional ideals with parameters under various specializations is different, which leads to choosing and checking the separating element, the key to computing the rational univariate representation, is difficult. In order to pick out the separating element, we first ensure that under each branch the ideal has the same number of zeros by partitioning the parameter space. Subsequently two ideas are given to choose and check the separating element. One idea is that by extending the subresultant theorem to parametric cases, we utilize the extended subresultant theorem to choose the separating element with the further partition of parameter space and then with the help of parametric greatest common divisor theory compute rational univariate representations. Another one is that we go straight to choose and check the separating element by the computation of parametric greatest common divisors, then immediately get the rational univariate representations. Based on these, we design two different algorithms for computing rational univariate representations of zero-dimensional ideals with parameters. Furthermore, the algorithms have been implemented on {\em Singular} and the performance comparison are presented.

\end{abstract}

\begin{keyword}
 Rational univariate representation, Parametric zero-dimensional ideal, Comprehensive Gr\"{o}bner system
\end{keyword}
\end{frontmatter}

\section{Introduction}\label{sec1}

 Solving multivariate polynomial equations has always been a classical algebraic problem, which plays an important role in many fields. The rapid development of computer technology makes the function of computer algebra system perfect. Thus, it provides a new opportunity for the solution of polynomial equations. Many new methods for solving polynomial equations have emerged in computer algebra and algebraic geometry, such as Wu's method \citep{Wu84}, \gr basis method \citep{Becker93}, resultant-based method \citep{Cox05} and eigenvalue methods \citep{Auzinger88,Stetter96}. These new theories and methods inject new vitality into the solution of polynomial equations, making them more and more widely used in computer-aided design, computer vision, cybernetics, robot trajectory design, curve and surface design and modeling.

 \cite{Rouillier99} proposed the rational univariate representation (RUR) to solve zero-dimensional polynomial systems and presented an efficient algorithm for computing rational univariate representations. The rational univariate representation consists in expressing all the coordinates of the roots for zero-dimensional ideal $I\subset k[x_1,\ldots,x_n]$ (where $k$ is a field of characteristic 0) as rational functions of the roots of an univariate polynomial. That is, the roots of $I$ can be represented in the following way:
	$$\V_L(I)=\left\{\left(\frac{g_1(\beta)}{g(\beta)},\ldots,\frac{g_n(\beta)}{g(\beta)}\right)\Big| \;\beta\in \V_L(\mathcal{X}(T))\right\},$$
 where $\mathcal{X},g,g_1,\ldots,g_n \in k[T]$ are univariate polynomials, $L$ is an algebraic closure of $k$, $\V_L(I)$ is the variety of $I$ in $L^n$. Since then, it has been extensively studied. \cite{Noro99} used modular method to compute the rational univariate representation of zero-dimensional ideal. \cite{Ouchi08} based on toric resultants gave the computation of the rational univariate representation. \cite{Tan2009} presented an improved algorithm for finding separating elements of zero-dimensional ideals. \cite{Zeng10} proposed a method for computing the rational univariate representation by Wu's method. \cite{Ma12} presented an approach to compute the rational univariate representation via properties of \gr basis.

 \cite{Tan09,Tan10} generalized the rational univariate representation theory to high-dimensional polynomial systems and proposed the rational representation theory. Along this, \cite{Shang17} proposed a simplified rational representation and \cite{Xiao2021} presented an improvement of the rational representation by introducing minimal Dickson basis proposed by \cite{Kapur13}. Similar to rational representation, \cite{Schost03} proposed parametric geometric resolution and \cite{Mohab18} used rational parametrizations to represent all irreducible components of real algebraic sets.

 Although the rational representation of polynomial systems has formed a relatively perfect theory, the research on the rational representation of polynomial systems with parameters as the method of solving parametric multivariate polynomial equations is still blank. In this paper, we extend the rational univariate representation of zero-dimensional ideals to the parametric case. Given an ideal $I=\langle f_1(U,X),\ldots,f_l(U,X)\rangle \subset k[U,X]$ with variables $X$ and parameters $U$. The considered problem is to find finite sets $E_i,N_i \subset k[U]$, $i\in \{1,\ldots,s\}$ such that for all $\bar{u} \in \V_L(E_i)\backslash \V_L(N_i)$, $I(\bar{u})=\langle f_1(\bar{u},X),\ldots,f_l(\bar{u},X)\rangle \subset L[X]$ is a zero-dimensional ideal and give an RUR of $I(\bar{u})$ to express the variety:
 $$\V_L(I(\bar{u}))=\left\{\left(\frac{g_{i1}(\bar{u},\beta)}{g_i(\bar{u},\beta)},\ldots,\frac{g_{in}(\bar{u},\beta)}{g_i(\bar{u},\beta)}\right)\Big|\;\beta \in \V_L(\mathcal{X}_i(\bar{u},T))\right\},$$
 where $g_{i}$, $g_{i1}$, $\ldots$, $g_{in}$, $\mathcal{X}_i \in k(U)[T]$.
 That is to compute a finite set:
 $$\left\{(E_1,N_1,\mathcal{X}_1,g_1,g_{11},\ldots,g_{1n}),\ldots,(E_s,N_s,\mathcal{X}_s,g_s,g_{s1},\ldots,g_{sn})\right\}.$$

 Two ideas based on the partition of parameter space are given. First, by means of comprehensive \gr systems (see \cite{Weispfenning92,Kalkbrener97,Montes02,Suzuki06,Nabeshima07,Kapur13}) and Finiteness theorem, we pick out the zero-dimensional branches which satisfy the ideal is zero-dimensional under parametric specializations. Different from the rational univariate representation of zero-dimensional ideals without parameters, the number of zeros for zero-dimensional ideals with parameters under various specializations is different, which leads to choosing and checking the separating element, the premise and the key to computing the rational univariate representation, is difficult. Thus, the second step is to determine the number of zeros by partitioning of the parameter space such that each branch of parameter space has the same number of zeros. Subsequently, one ideas is that we directly use extended subresultant theorem \citep{MB1993} for parametric cases to choose the separating element corresponding to each branch with the further partition of parameter space and then with the help of parametric greatest common divisor theory \citep{Nagasaka17,Kapur18} compute a finite set of which each branch shares the same expression of rational univariate representation. Another one is that we go straight to check the separating elements and get the rational univariate representations by the computation of parametric greatest common divisors, which actually uses the extended subresultant theorem potentially. As a result, we present two algorithms for computing rational univariate representations of zero-dimensional ideals with parameters.

 This paper is an extension of \cite{Wang22}, and new contributions are as follows. 1) Unlike previous study which divides the parameter space by directly utilizing the extended subresultant theorem for parametric cases for choosing the separating element, we give a new idea to choose and check the separating element based on parametric greatest common divisors, then present the different algorithm for computing rational univariate representations of zero-dimensional ideals with parameters. 2) The two algorithms (the algorithm \citep{Wang22} and new algorithm) have been implemented on the computer algebra system {\em Singular}. Experimental data and their performance comparison are reported in the paper. Moreover, the codes and examples are available on the web: \url{http://www.mmrc.iss.ac.cn/~dwang/software.html}.

 The rest of the paper is organized as follows.In Section 2, we introduce some notations and definitions, and extend the subresultant theorem to parametric cases. The main content is presented in Section 3. We use comprehensive \gr systems, the extended subresultant theorem and parametric greatest common divisors to give the rational univariate representation of zero-dimensional ideals with parameters, and two algorithm are presented. In Section 4, we give an example to illustrate the algorithms. The implementation of the algorithms and the performance comparison are presented in Section 5. We end with some concluding remarks in Section 6.

\section{Preliminaries}\label{sec2}

 In this section we will introduce some notations and definitions to prepare for the discussion of this article.

 Let $k$ be a field of characteristic 0, and $L$ be its algebraic closure. $k[X]$ is the polynomial ring over $k$ in the variables $X=\{x_1,\ldots,x_n\}$ and $k[U, X]$ is the parametric polynomial ring with the parameters $U=\{u_1,\ldots,u_m\}$ and variables $X$. $I$ is an ideal of $k[X]$ and $\V_L(I)$ is the variety of $I$ in $L^n$. For a set $N \subset k[X]$, $p \in k[X]$, we have the notation $N \times p = \{f\cdot p \;|\; f\in N\}$.

 \subsection{Rational univariate representation}\label{sec2.1}
 Let $I \subset k[X]$ be a zero-dimensional ideal. Since $I$ is zero-dimensional, $k[X]/I$ is a linear space over $k$ by Finiteness theorem. For all $t \in k[X]$, we denote by $m_t$ the $k-$linear map:
    \begin{align*}
    m_t:k[X]/I & \rightarrow k[X]/I\\
    \overline{f} &\mapsto \overline{tf},
    \end{align*}
 where $\overline{f}$ denotes the residue class in $k[X]/I$ of any polynomial $f \in k[X]$.

 We denote by $M_t$ the matrix representation of $m_t$ w.r.t. a basis in quotient ring $k[X]/I$. And we call $m_t$ the multiplication map and $M_t$ the multiplication matrix.

 \begin{theorem}[\cite{Rouillier99}]\label{thm1}
 Let $I \subset k[X]$ be a zero-dimensional ideal, $t \in k[X]$ and $m_t$ be the multiplication map. Then the eigenvalues of $M_t$ are $\{t(p):p \in \V_L(I)\}$. More specifically, the characteristic polynomial of $m_t$ is
 $$\mathcal{X}_t(T)=\prod_{p \in \V_L(I)}(T-t(p))^{\mu(p)},$$
 where $\mu(p)$ is the multiplicity of $p$ in $\V_L(I)$.
 \end{theorem}

 Similarly, given the polynomial $h \in k[X]$, we can construct a bilinear form
    $$S_h(f,g)=\text{Tr}(m_{hf} \cdot m_{g})=\text{Tr}(m_{hfg}).$$

 Suppose that $B=\{X^{\alpha_1},\ldots, X^{\alpha_s}\}$ is a basis of $k[X]/I$. Let $Q_h$ be the matrix of $S_h$ w.r.t. $B$, i.e., $Q_h=(S_h(X^{\alpha_i},X^{\alpha_j}))_{1\leq i,j \leq s}\in k^{s \times s}$. $Q_h$ is called the Hermite's quadratic form associated $h$ w.r.t. $B$ in quotient ring $k[X]/I$.

    \begin{theorem}[\cite{Cox06}, Page 71, Theorem 5.2]\label{thm2}
    Let $I \subset k[X]$ be a zero-dimensional ideal. Then the number of points in $\V_L(I)$ is equal to the rank of $Q_1$.
    \end{theorem}

 Next, we review the definition of separating elements which plays an important role in rational univariate representation theory.
 \begin{definition}\label{def1}
   Let $I \subset k[X]$ be a zero-dimensional ideal. A polynomial $t \in k[X]$ separates $\V_L(I)$, if
   $$p_1,p_2 \in \V_L(I), p_1 \neq p_2 \Rightarrow t(p_1)\neq t(p_2).$$
   We also call $t$ a separating element of $I$.	
 \end{definition}

 \begin{definition}
  Let $I \subset k[X]$ be a zero-dimensional ideal, and $t\in k[X]$. $\mathcal{X}_t$ is the characteristic polynomial of multiplication map $m_t$.
  For any $v \in k[X]$, we define:
 \[g(v,T)=\sum\limits_{\alpha\in \V_{\k}(I)}\mu(\alpha)v(\alpha)\prod\limits_{y\neq t(\alpha),y\in \V_{\k}(\mathcal{X}_t)}(T-y).\]

  For any $t \in k[X]$, the $t$-representation of $I$ is the $(n+2)$-tuple:
  \[\{\mathcal{X}_t(T),g(1,T),g(x_1,T),\ldots,g(x_n,T)\}.\]

  If $t$ separates $\V_{\k}(I)$, the $t$-representation of $I$ is called the {\bf{Rational Univariate Representation}} (RUR) of $I$ associated to $t$.
 \end{definition}

 In \cite{Rouillier99}, Rouillier proposed an algorithm to compute the RUR of a zero-dimensional ideal.

\begin{center}
\textbf{Algorithm for Computing the RUR}
\end{center}

{\bf Input:} $I=\langle f_1,\ldots,f_l\rangle \subset k[X]$ is a zero-dimensional ideal.

\begin{enumerate}
\item Compute the \gr basis of $I$ with a monomial order $\prec$ and a basis $B$ of $k[X]/I$ over $k$.
\item Compute the rank of $Q_1$ with respect to $B$.
\item Set $d=\text{rank}(Q_1)$ and choose $t \in S=\{x_{1}+i x_{2}+\cdots+i^{n-1} x_{n},$ $ i=1, \ldots, nd(d-1) / 2\}$ randomly, then compute the characteristic polynomial $\mathcal{X}_t$ of $m_t$.
\item Let $\overline{\mathcal{X}_{t}}=\mathcal{X}_{t} / \gcd \left(\mathcal{X}_{t}^{\prime}, \mathcal{X}_{t}\right)$. If $\deg(\overline{\mathcal{X}_{t}}) \neq \text{rank}(Q_1)$, then goto (3).
\item Suppose that $\overline{\mathcal{X}_{t}}=\sum_{j=0}^{d}a_jT^{d-j}$. Compute
\begin{align}
g(T)=&\sum_{i=0}^{d-1}\sum_{j=0}^{d-i-1}\text{Tr}(m_{t^i})a_jT^{d-i-j-1}, \label{eq1}\\
g_k(T)=&\sum_{i=0}^{d-1}\sum_{j=0}^{d-i-1}\text{Tr}(m_{x_kt^i})a_jT^{d-i-j-1},k=1,\ldots,n. \label{eq2}
\end{align}
\end{enumerate}

{\bf Output:} $\left\{\mathcal{X}_{t}(T), g(T), g_1(T), \ldots, g_n(T)\right\} \subset k[T]$.

By the above algorithm, the variety of $I$ can be represented as
$$\V_L(I)=\left\{\left(\frac{g_1(\beta)}{g(\beta)},\ldots,\frac{g_n(\beta)}{g(\beta)}\right)\Big|\;\beta \in \V_L(\mathcal{X}_t(T))\right\}.$$

There are some remarks:
\begin{itemize}
	\item  The Step (4) is aimed to check whether $t$ is a separating element of $I$. According to Definition \ref{def1}, Theorem \ref{thm1} and Theorem \ref{thm2}, $t$ is a separating element of $I$ if and only if $\deg(\overline{\mathcal{X}_{t}})=\text{rank}(Q_1)$. In \cite{Rouillier99}, it has been proved that $S$ contains at least one element that separates $\V_L(I)$.
	\item We can use Formula (\ref{eq1}) and (\ref{eq2}) to compute $g(T)$ and $g_i(T)$ only if $t$ is a separating element.
	\item The RUR of $I$ computed by this algorithm is not unique. It depends on the choice of the separating element.
\end{itemize}

\subsection{Comprehensive \gr System and Parametric Greatest Common Divisor}
 In order to deal with the polynomial system with parameters, we will introduce the comprehensive \gr system. Fix a monomial order $\prec_X$ in $k[X]$ and monomial order $\prec_U$ in $k[U]$. Let $\prec_{X,U}$ be a block order, $U \prec \prec X$, within $\prec_U$ and $\prec_X$. For convenience, we often omit writing the order if there is no confusion.

\begin{definition}[\cite{Weispfenning92}]
	Let $I=\langle f_1,\ldots,f_l\rangle \subset k[U,X]$, $S$ be a subset of $L^m$, $G_i =\{g_{i1},\ldots, g_{it}\}\subset k[U, X]$ and $E_i,N_i \subset k[U]$ for $i=1,\ldots, s$ such that $S = \bigcup_{i=1}^{l}\V_L(E_i)\backslash \V_L(N_i)$. A finite set
$$\mathcal{G}=\{(E_1,N_1,G_1),\ldots,(E_s,N_s,G_s)\}$$
is called a comprehensive \gr system (CGS) on $S$ for $I$ if for every $i \in \{1,\ldots, s\}$, $G_i(\bar{u})=\{g_{i1}(\bar{u},X),\ldots, g_{it}(\bar{u},X)\}$ is a \gr basis of $I(\bar{u})=\langle f_1(\bar{u},X),\ldots,f_l(\bar{u},X)\rangle \subset L[X]$ for $\bar{u}\in \V_L(E_i)\backslash \V_L(N_i)$. Each $(E_i,N_i,G_i)$ is called a branch of $\mathcal{G}$. In particular, if $S=L^m$, then $\mathcal{G}$ is called a comprehensive \gr system for $I$.
\end{definition}

\begin{definition}
A comprehensive Gröbner system $\mathcal{G}=\{(E_{1},N_1,G_{1}),$ $\ldots,(E_{s}, N_s,G_{s})\}$ on $S$ for $I$ is said to be minimal, if for each $i=1, \ldots, s$,
	\begin{enumerate}
		\item $\V_L(E_{i})\backslash \V_L(N_{i}) \neq \emptyset$, and furthermore, whenever $i \neq j$,
		$$\V_L(E_{i})\backslash \V_L(N_{i})\cap \V_L(E_{j})\backslash \V_L(N_{j})=\emptyset;$$
		\item $G_{i}(\bar{u})$ is a minimal Gröbner basis for $I(\bar{u}) \subset L[X]$
		for $\bar{u} \in \V_L(E_{i})\backslash \V_L(N_{i})$;
		\item for each $g\in G_{i}$, assume that $c_1={\rm LC}_X(g)$ is the leading coefficient of $g$ w.r.t. $X$, then $c_1(\bar{u})\neq 0$ for any $\bar{u} \in \V_L(E_i)\backslash \V_L(N_i)$.
		\end{enumerate}
\end{definition}

For computing the minimal comprehensive \gr system, we can refer to the algorithm proposed by \citet{Kapur13}.

Based on the comprehensive \gr system, \citet{Kapur18} proposed an efficient algorithm for computing the greatest common divisor of polynomials with parameters.

\begin{definition}
 	For $F=\left\{f_{1}, \ldots, f_{l}\right\} \subset k[U][X]$ and $S \subset L^{m}$, we call $\left\{\left(E_{1}, N_{1}, g_{1}\right), \ldots,\left(E_{r}, N_{r}, g_{r}\right)\right\}$ a parametric GCD of $F$ on $S$, if for each $i=1, \ldots, r$, $g_i(\bar{u},X) \in L[X]$ is a greatest common divisor of $\{f_{1}(\bar{u},X), \ldots, f_{l}(\bar{u},X)\} \subset L[X]$ for  $\bar{u} \in \V_L(E_i)\backslash \V_L(N_i)$, where $E_i,N_i \subset k[U]$, $g_{i} \in k[U][X]$ for $i=1,\ldots, r$ and $S=\cup_{i=1}^{r} \V_L(E_i)\backslash \V_L(N_i)$. If $S=k^{m}$, we simply call it a parametric $\mathrm{GCD}$ of $F$.
 \end{definition}

\subsection{Subresultant and Common Divisor}
  In the process of computing the RUR, we can see that it is necessary to compute $\gcd(\mathcal{X}_{t}^{\prime},\mathcal{X}_{t})$. Therefore, if the polynomial system is with parameters, we need to discuss the relationship between $\gcd(\mathcal{X}_{t}^{\prime},\mathcal{X}_{t})$ and the value of parameters. In the following, we will build the connection between the degree of $\gcd(\mathcal{X}_{t}^{\prime},\mathcal{X}_{t})$ and the parameters $U$ by using the subresultant theory.

  First, we review the theorem of the subresultant.

 \begin{theorem}[\cite{MB1993}]\label{thm3}
 Let $S$ be a unique factorization domain with identity, and $A(x), B(x) \in S[x]$ be univariate polynomials of positive degrees $m$ and $n$, respectively. Then for all $0 \leq i<\min (m, n)$, the following three statements are equivalent:
 \begin{enumerate}
		\item $A(x)$ and $B(x)$ have a common divisor with degree $>i$;
		\item $(\forall j \leq i)\left[\operatorname{SubRes}_{j}(A, B)=0\right]$;
		\item $(\forall j \leq i)\left[\operatorname{PSC}_{j}(A, B)=0\right]$.
 \end{enumerate}
 where $SubRes_{j}(A, B)$ denotes the $i$th subresultant of $A$ and $B$, ${PSC}_{j}(A, B)$ denotes the leading coefficient of $SubRes_{j}(A, B)$.
 \end{theorem}

 Now we extend this result to the case for parameters and obtain the relationship between the degree of $\gcd(\mathcal{X}_{t}^{\prime},\mathcal{X}_{t})$ and the values of parameters.

 \begin{lemma}\label{cor1}
 Let $\mathcal{X}(U,T)=\sum_{i=0}^{l}c_i(U)T^i \in k[U,T]$ and $E,N \subset k[U]$. Suppose that for $\bar{u} \in \V_L(E) \backslash \V_L(N)$, $c_l(\bar{u})\neq 0$, then there exist $E_0, \ldots, E_{l-1}, N_0, \ldots, N_{l-1} \subset k[U]$ such that for each $0 \leq i \leq l-1$ and for all $\bar{u} \in \V_L(E) \backslash \V_L(N)$,
	$$\deg\left(\gcd(\mathcal{X}(\bar{u},T),\mathcal{X}^{\prime}(\bar{u},T))\right)=i \iff \bar{u} \in \V_L(E_i) \backslash \V_L(N_i).$$
 where $\mathcal{X}' \in k[U,T]$ is the derivation of $\mathcal{X}$ w.r.t. $T$.
 \end{lemma}

 \begin{proof}
 We regard $\mathcal{X}$ and $\mathcal{X}^\prime$ as polynomials in $S[T]$, where $S=k[U]$. Then $\operatorname{PSC}_{j}(\mathcal{X},\mathcal{X}^\prime) \in k[U]$. For $\bar{u} \in \V_L(E) \backslash \V_L(N)$, $c_l(\bar{u})\neq 0$, so
 $$\deg_T(\mathcal{X}(U,T))=\deg(\mathcal{X}(\bar{u},T))~\text{and} \deg_T(\mathcal{X}^\prime(U,T))=\deg(\mathcal{X}^\prime(\bar{u},T)).$$
 It implies that the degree of $\mathcal{X}(U,T)$ and $\mathcal{X}'(U,T)$ w.r.t. $T$ does not change under $\bar{u}\in \V_L(E) \backslash \V_L(N)$. By the definition, for $\bar{u} \in \V_L(E) \backslash \V_L(N)$, $$\operatorname{PSC}_{j}(\mathcal{X},\mathcal{X}^\prime)(\bar{u})=\operatorname{PSC}_{j}(\mathcal{X}(\bar{u},T),\mathcal{X}^\prime(\bar{u},T)), j=0,1,\ldots,l-1.$$

 Let $\lambda=l-1$ and
	$$E_i=E \cup \{\operatorname{PSC}_{j}(\mathcal{X},\mathcal{X}^\prime)|\;j<i\},\;0 \leq i \leq \lambda,$$
	$$N_i=N \times \{\operatorname{PSC}_{i}(\mathcal{X},\mathcal{X}^\prime)\}, \;0 \leq i < \lambda,\text{ and }N_\lambda=N,$$
	where $N \times \{f\}=\{gf |\; g\in N\}$). By Theorem \ref{thm3}, for each $0 \leq i<\lambda$ and  $\bar{u} \in \V_L(E) \backslash \V_L(N)$, the following statements are equivalent:
	\begin{enumerate}
		\item $\deg\left(\gcd(\mathcal{X}(\bar{u},T),\mathcal{X}^{\prime}(\bar{u},T))\right)=i$;
		\item $(\forall j <i)\;[\operatorname{PSC}_{j}(\mathcal{X}(\bar{u},T),\mathcal{X}^\prime(\bar{u},T))=0]$ and $\operatorname{PSC}_{i}(\mathcal{X}(\bar{u},T),$ $\mathcal{X}^\prime(\bar{u},T)) \neq 0$;
		\item $(\forall j <i)\;[\operatorname{PSC}_{j}(\mathcal{X},\mathcal{X}^\prime)(\bar{u})=0]$ and $\operatorname{PSC}_{i}(\mathcal{X},\mathcal{X}^\prime)(\bar{u}) \neq 0$;
		\item $\bar{u} \in \V_L(E_i) \backslash \V_L(N_i)$.
	\end{enumerate}
    Moreover, it is easy to verify that the conclusion is also true when $i=\lambda$.
\end{proof}

\section{RUR of zero-dimensional ideal with parameters}\label{sec3}
 In the section, we will consider the RUR of zero-dimensional ideal with parameters $U$. Let $I= \langle f_1(U,X),\ldots,f_l(U,X) \rangle \in k[U,X]$ and $E,N \subset k[U]$. For $f \in k[U]$, $N \times \{f\}=\{gf |\; g\in N\}$. For every $\bar{u} \in \V_L(E)\backslash \V_L(N)$, the ideal generated by $f_1(\bar{u},X),\ldots,$ $f_l(\bar{u},X)$ is denoted by $I(\bar{u})$. We want to find finite sets $E_i,N_i \subset k[U]$, $i\in \{1,\ldots,s\}$ such that for $\bar{u} \in \V_L(E_i)\backslash \V_L(N_i)$, $I(\bar{u})$ is a zero-dimensional ideal and give an RUR of $I(\bar{u})$ to express
 $$\V_L(I(\bar{u}))=\left\{\left(\frac{g_{i1}(\bar{u},\beta)}{g_i(\bar{u},\beta)},\ldots,\frac{g_{in}(\bar{u},\beta)}{g_i(\bar{u},\beta)}\right)\Big|\;\beta \in \V_L(\mathcal{X}_i(\bar{u},T))\right\},$$
 where $g_{i}$, $g_{i1}$, $\ldots$, $g_{in}$, $\mathcal{X}_i \in k(U)[T]$.

 The ideas are based on the partition of parameter space. The first idea contains four steps and another contains three steps, among which the first two steps are same.

 The first idea is as follows.

 {\itshape {\bf Step 1.~}} Pick zero-dimensional branches. That is to compute $E_i,N_i \subset k[U]$, $i=1,\ldots,s$, such that $I(\bar{u})$ is a zero-dimensional ideal if and only if $\bar{u}$ belongs to some $\V_L(E_i)\backslash \V_L(N_i)$.

 {\itshape {\bf Step 2.~}} Determine the number of zeros. For each $E_i,N_i$ in Step 1, compute $E_{ij},N_{ij} \subset k[U]$ such that $\V_L(E_{i})\backslash \V_L(N_{i})=\bigcup_{j=1}^{s_1}\V_L(E_{ij})\backslash$ $\V_L(N_{ij})$ and for each $j \in \{1,\ldots,s_1\}$, $I(\bar{u})$ has the same number of zeros for all $\bar{u}\in \V_L(E_{ij})\backslash \V_L(N_{ij})$.

 {\itshape {\bf Step 3.}} Choose and check the separating element. For each $E_{ij},N_{ij}$ in Step 2, compute $E_{ijk},N_{ijk} \subset k[U]$ such that $\V_L(E_{ij})\backslash$ $\V_L(N_{ij})$ $=\bigcup_{k=1}^{s_2}\V_L(E_{ijk})\backslash\V_L(N_{ijk})$ and for each $k \in \{1,\ldots,s_2\}$, $I(\bar{u})$ shares the same separating element for all $\bar{u}\in \V_L(E_{ijk})\backslash \V_L(N_{ijk})$.

 {\itshape {\bf Step 4.~}} Give the RUR for each branch. For each $E_{ijk},N_{ijk}$ in Step 3, compute $E_{ijkl},N_{ijkl} \subset k[U]$ such that $\V_L(E_{ijk}) \backslash \V_L(N_{ijk})=\bigcup_{l=1}^{s_3}\V_L(E_{ijkl})\backslash \V_L(N_{ijkl})$ and for each $l \in \{1,\ldots,s_3\}$, for all $\bar{u}\in \V_L(E_{ijkl})\backslash \V_L(N_{ijkl})$, $I(\bar{u})$ shares the same expression of rational univariate representation. 	

\vh
 The second idea is to skip the third step and go straight to check the separating elements and get the RURs by computing GCDs of $\mathcal{X}_t$ and $\mathcal{X}_t^{\prime}$. That is to say,

 {\itshape {\bf Step 3$^\spadesuit$.}} Choose the separating element and compute GCDs of $\mathcal{X}_t$ and $\mathcal{X}_t^{\prime}$. For each $E_{ij},N_{ij}$ in Step 2, compute $E_{ijk},N_{ijk} \subset k[U]$ such that $\V_L(E_{ij})\backslash$ $\V_L(N_{ij})$ $=\bigcup_{k=1}^{s_2}\V_L(E_{ijk})\backslash\V_L(N_{ijk})$, and for each $k \in \{1,\ldots,s_2\}$, the GCD of $\mathcal{X}_t$ and $\mathcal{X}^\prime_t$ has the same expression, and $I(\bar{u})$ shares the same separating element as well as the same expression of rational univariate representation for all $\bar{u}\in \V_L(E_{ijk})\backslash \V_L(N_{ijk})$.

\subsection{Algorithm 1}

 Based on four steps described above, we can design the first algorithm to compute RURs of zero-dimensional ideal with parameters. Next, we will explain these steps in detail, and then present the algorithm.

\subsubsection{Partition according to the \gr basis}
  In fact, Step 1 is easy to be achieved by comprehensive \gr system and Finiteness Theorem.

  \begin{theorem}[Finiteness Theorem \cite{Cox06}]
  	Let $I \subset k\left[X\right]$ be an ideal. Then the following conditions are equivalent:
  	\begin{enumerate}
  		\item The algebra $A=k[x_{1}, \ldots, x_{n}] / I$ is finite-dimensional over $k$.
  		\item The variety $\V_L(I) \subset L^{n}$ is a finite set.
  		\item If $G$ is a \gr basis for $I$, then for each $1 \leq i \leq n$, there is an $m_{i} \geq 0$ such that $x_{i}^{m_{i}}=\rm{LM}(g)$ for some $g \in G$, where $\operatorname{LM}(g)$ is the leading monomial of $g$.
  	\end{enumerate}
  \end{theorem}

  An ideal satisfying any of the above conditions is said to be zero-dimensional.

 \begin{theorem}\label{thm3.2}
  	Let $I=\langle f_1(U,X),\ldots,f_l(U,X)\rangle \subset k[U,X]$. There exists a finite set
  	$\{(E_1,N_1,G_1),\ldots,$ $(E_s,N_s,G_s)\},$
  	such that:
  	\begin{enumerate}
  		\item $I(\bar{u})$ is a zero-dimensional ideal if and only if $$\bar{u}\in \bigcup_{i=1}^s\V_L(E_i) \backslash \V_L(N_i).$$
  		\item For $\bar{u} \in \V_L(E_i) \backslash \V_L(N_i)$, $G_i(\bar{u})=\{g_{i1}(\bar{u},X),\ldots, g_{ie}(\bar{u},X)\}$ is a \gr basis of $I(\bar{u})$.
  	\end{enumerate}
  \end{theorem}

 \begin{proof}
 We only need to compute the minimal comprehensive \gr system $\mathcal{G}=\{(E_{1},N_1,G_{1}),$ $\ldots,(E_{s_0}, N_{s_0},G_{s_0}\}$ for $I$. Choose the branches $\left\{\left(E_{i_1},N_{i_1},G_{i_1}\right), \ldots,\left(E_{i_{l}}, N_{i_{l}},G_{i_{l}}\right)\right\}$ which satisfy for each $1\leq j \leq n$, $k[x_j] \cap \text{LM}_X(G_{i_v}) \neq \emptyset$ $(v=1,\ldots,l)$. By Finiteness Theorem and the definition of minimal comprehensive \gr system, $\left\{\left(E_{i_1},N_{i_1},G_{i_1}\right), \ldots,\left(E_{i_{l}}, N_{i_{l}},G_{i_{l}}\right)\right\}$ satisfies the conclusion.
\end{proof}

 \begin{remark}\label{rem3.3}
 For each branch $(E_i,N_i,G_i)$ in Theorem \ref{thm3.2}, and for $\bar{u} \in \V_L(E_i) \backslash \V_L(N_i)$,
  $${\rm{LM}}(G_i(\bar{u}))=\{{\rm{LM}}(g_{i1}(\bar{u},X)),\ldots,{\rm{LM}}(g_{ie}(\bar{u},X))\}.$$
 Therefore, for all $\bar{u} \in \V_L(E_i) \backslash \V_L(N_i)$, $L[X]/I(\bar{u})$ shares the same basis:
	$$B_i=\{X^{\alpha_1},\;\ldots,\;X^{\alpha_r}|\;X^{\alpha_j} \notin \langle {\rm{LM}}_{X}(G_i) \rangle, j=1,\ldots,r\},$$
	where ${\rm{LM}}_{X}(G)$ is a set of the leading monomials of all elements in $G$ w.r.t. $X$.
 \end{remark}

\subsubsection{Partition according to the number of zeros}

 Actually, Step 2 is a preparation for Step 3. For non-parametric cases, $t \in k[X]$ is a separating element of $I$ if and only if $\deg(\overline{\mathcal{X}_{t}}) = \text{rank}(Q_1)$, where $\overline{\mathcal{X}_{t}}=\mathcal{X}_{t} / \gcd \left(\mathcal{X}_{t}^{\prime}, \mathcal{X}_{t}\right)$. As for parametric cases, $Q_1$ and $\mathcal{X}_{t}$ are both with parameters. Therefore, choosing a $t \in k[X]$,  it should discuss when dose $\deg(\overline{\mathcal{X}_{t}}) = \text{rank}(Q_1)$ hold with the change of the parameters. First, we need to determine $\text{rank}(Q_1)$, equally the number of zeros for $I$, under parametric specializations.

 Let $(E,N,G)$ be one of the branches obtained in Theorem \ref{thm3.2} and $G=\{g_1,\ldots,g_e\} \subset k[U,X]$. Then $(E,N,G)$ satisfies that for any $\bar{u} \in \V_L(E)\backslash \V_L(N)$,
 \begin{enumerate}
		\item $I(\bar{u})$ is zero-dimensional and $G(\bar{u})$ is a \gr basis of $I(\bar{u})$;
		\item for each $g\in G$, $c_1(\bar{u})\neq 0$, where $c_1$ is the leading coefficient of $g$ w.r.t. $X$.
 \end{enumerate}

 Assume that $B=\{X^{\alpha_1},\ldots,X^{\alpha_r}\}$ be a basis of $L[X]/I(\bar{u})$. For $t \in k[X]$, let $R_i$ be a remainder of $t\cdot X^{\alpha_1}$ on division by $G$ in $k(U)[X]$, i.e.,
     $$t\cdot X^{\alpha_i}=\sum_{j=1}^{e}q_{ij}g_j+R_i,\;i=1,\ldots,r.$$
 where $q_{ij},R_i \in k(U)[X]$ and their denominators are the products of some factors of the leading coefficient of $g_j\in G$  w.r.t. $X$. So they are not equal to zero under $\bar{u} \in \V_L(E)\backslash \V_L(N)$. Therefore, we have
     $$t\cdot X^{\alpha_i}=\sum_{j=1}^{e}q_{ij}(\bar{u},X)g_j(\bar{u},X)+R_i(\bar{u},X),\;i=1,\ldots,r.$$
 The remainder of $t\cdot X^{\alpha_1}$ on division by $G(\bar{u})$ in $L[X]$ is equal to $R_i(\bar{u},X)$. Let $M_t^I \in k(U)^{r \times r}$ be the matrix satisfying
$$(R_1,\ldots,R_r)=B \cdot M_t^I,$$
 For any $\bar{u} \in \V_L(E)\backslash \V_L(N)$,
$$(R_1(\bar{u},X),\ldots,R_r(\bar{u},X))=B \cdot M_t^I(\bar{u}).$$
 Assume that $M_t^{I(\bar{u})}$ is the multiplication matrix of $t$ w.r.t. $B$ in quotient ring $L[X]/I(\bar{u})$. Then $M_t^I(\bar{u})=M_t^{I(\bar{u})}$. Therefore, we obtain the following lemma.

 \begin{lemma}\label{lem3.4}
 Let $(E,N,G)$, $B$ and $M^I_t \in k(U)^{r \times r}$ be as defined above. For any $\bar{u} \in \V_L(E)\backslash \V_L(N)$,
 \begin{enumerate}
 \item The multiplication matrix $M_t^{I(\bar{u})}$ of $t$ w.r.t. $B$ in quotient ring $L[X]/I(\bar{u})$ is equal to $M^I_t(\bar{u})$.
 \item Assume that the characteristic polynomial of $M^I_t$ over $k(U)$ is $\mathcal{X}_{t}(U,T)$, then the characteristic polynomial of $M^{I(\bar{u})}_t$ over $L$ is $\mathcal{X}_{t}(\bar{u},T)$.
 \item Let $h \in k[X]$,  $Q_h^I \triangleq (\text{Tr}(M^I_{hX^{\alpha_i}X^{\alpha_j}}))_{1\leq i,j\leq r}$ be a $r \times r$ matrix over $k(U)$. Then the Hermite's quadratic form associated $h$ w.r.t. $B$ in quotient ring $L[X]/I(\bar{u})$ is equal to $Q_h^I(\bar{u}).$
 \end{enumerate}
 \end{lemma}	

  By the above analysis, first we compute $Q_1^I \in k(U)^{r \times r}$ defined above. For all $\bar{u}\in \V_L(E)\backslash \V_L(N)$, $Q_1^{I(u)}=Q_1^{I}(u)$ by Lemma \ref{lem3.4} and the number of points in $\V_L(I(\bar{u}))$ is equal to $\text{rank}(Q_1^{I(u)})$ by Theorem \ref{thm2}.

 \begin{theorem}\label{thm22}
	Let $(E,N,G)$ be one of the branches in Theorem \ref{thm3.2}. Then there exists a finite set
	$$\{(E_1,N_1,k_1),\ldots,(E_s,N_s,k_s)\}, ~E_i, ~N_i \subset k[U],\;k_i \in \mathbb{Z}_{>0},$$
	such that $\V_L(E) \backslash \V_L(N)=\cup_{i=1}^s\V_L(E_i) \backslash \V_L(N_i)$ and for each $i \in \{1,\ldots,s\}$ and all $\bar{u} \in \V_L(E_i)\backslash \V_L(N_i)$, the number of points in $\V_L(I(\bar{u}))$ is equal to $k_i$.
 \end{theorem}

\begin{proof}
    We first compute $Q_1^I \in k(U)^{r \times r}$ as defined above. It suffices to discuss $\text{rank}(Q_1^I(\bar{u}))$ for $\bar{u} \in \V_L(E)\backslash \V_L(N)$. Let $D$ is the least common multiple of all denominators of entries in $Q_1^I$. Then $D \in k[U]$ and $D(\bar{u}) \neq 0$ for any $\bar{u} \in \V_L(E)\backslash \V_L(N)$. Let $Q=D\cdot Q_1^I \in k[U]^{r \times r}$. It suffices to discuss $\text{rank}(Q(\bar{u}))$.

    Assume that $Q=(f_{ij})_{1\leq i,j\leq r}$. Using the method in \cite{Ma17}, we construct polynomial
    $f_i=\sum_{j=1}^rf_{ij}z_j \in k[U][Z], Z=\{z_1,\ldots,z_r\}$, $i=1,\ldots,r$. Compute the minimal comprehensive \gr system $\mathcal{G}=\{(E_1,N_1,G_1),\ldots,(E_s,N_s,G_s)\}$ of $F=\{f_1,\ldots,f_r\}$ where $U$ is the parameters. By Theorem 1 and Theorem 2 in \cite{Ma17}, for any $\bar{u} \in \V_L(E_i) \backslash \V_L(N_i)$, the rank of $Q(\bar{u})$ is equal to the number of elements in $G_i$. Let $k_i$ denote the number of elements in $G_i$, $\{(E_1,N_1,k_1),\ldots,(E_s,N_s,k_s)\}$ satisfies the conditions we need. 	
\end{proof}

\subsubsection{Partition according to separating elements}
 Choosing the separating element is the key to obtain the rational univariate representation. In the following, we will analyze each branch in Theorem \ref{thm22}. Assume that for all $\bar{u} \in \V_L(E) \backslash \V_L(N)$, the number of points in $\V_L(I(\bar{u}))$ is the same, denoted by $k_0$.

 If $k_0=1$, we can choose $t=1$ which is a separating element of $I(\bar{u})$ for $\bar{u} \in \V_L(E) \backslash \V_L(N)$. Otherwise, choose randomly a $t \in \{x_{1}+i x_{2}+\cdots+i^{n-1} x_{n}, i=1 \ldots n k_0(k_0-1) / 2\}$. $\V_L(E) \backslash \V_L(N)$ can be divided into two parts, $\V_L(E_i) \backslash \V_L(N_i)$, $i=1,2$, where
\begin{enumerate}
	\item $t$ is a separating element of $I(\bar{u})$ for $\bar{u} \in \V_L(E_1) \backslash \V_L(N_1)$.
	\item $t$ is not a separating element of $I(\bar{u})$ for $\bar{u} \in \V_L(E_2) \backslash \V_L(N_2)$.
\end{enumerate}

We will prove in fact the second part is "small enough" in numbers of points. Thus we first find a separating element for the most part of $\V_L(E) \backslash \V_L(N)$.  For the small part, we will choose another $t \in k[X]$. Sequentially, it does a recursive procedure.

\begin{theorem}\label{thm33}
Let $I=\langle f_1(U,X),\ldots,f_l(U,X)\rangle \subset k[U,X]$, $E,N \subset k[U]$, $G \subset k[U,X]$. Assume that for all $\bar{u} \in \V_L(E)\backslash \V_L(N)$,
\begin{enumerate}
	\item $I(\bar{u})$ is zero-dimensional and $G(\bar{u})$ is a \gr basis of $I(\bar{u})$;
	\item for each $g\in G$, $c_1(\bar{u})\neq 0$, where $c_1$ is the leading coefficient of $g$ w.r.t. $X$;
	\item the number of points in $\V_L(I(\bar{u}))$ is equal to $k_0\neq 1$.
\end{enumerate}
 Suppose that $t \in k[X]$, $\mathcal{X}_t(U,T)\in k(U)[T]$ is the characteristic polynomial as defined in Lemma \ref{lem3.4}. Let $\mathcal{X}(U,T)\in k[U][T]$ (briefly, $\mathcal{X}$)be the product of $\mathcal{X}_t(U,T)$ and the least common multiple of all denominators of $\mathcal{X}_t(U,T)$.  Then $t$ is separating element of $I(\bar{u})$ if and only if $\bar{u} \in \V_L(E)\backslash \V_L(N \times \{\operatorname{PSC}_{d}(\mathcal{X},\mathcal{X}^\prime)\})\neq\emptyset$.
\end{theorem}

\begin{proof}
 By Definition \ref{def1} and Theorem \ref{thm1}, for $\bar{u} \in \V_L(E) \backslash \V_L(N)$, $t$ is a separating element of $I(\bar{u})$ if and only
 $$\deg\left(\frac{\mathcal{X}(\bar{u},T)}{\gcd (\mathcal{X}^{\prime}(\bar{u},T), \mathcal{X}(\bar{u},T))} \right)=k_0,$$
 which is equivalent to
 $$\deg(\gcd \left(\mathcal{X}^{\prime}(\bar{u},T), \mathcal{X}(\bar{u},T)\right))=\deg (\mathcal{X}(\bar{u},T)) -k_0.$$
 Since $\deg_T (\mathcal{X})=\deg(\mathcal{X}(\bar{u},T))$, the above equation becomes
 $$\deg(\gcd \left(\mathcal{X}^{\prime}(\bar{u},T), \mathcal{X}(\bar{u},T)\right))=\deg_T (\mathcal{X}) -k_0.$$
 Let $d=\deg_T (\mathcal{X})-k_0$. By Theorem \ref{thm1}, for any $\bar{u} \in \V_L(E) \backslash \V_L(N)$, $\mathcal{X}(\bar{u},T)$ has at most $k_0$ distinct roots.
 It implies that
	$$\deg(\gcd \left(\mathcal{X}^{\prime}(\bar{u},T), \mathcal{X}(\bar{u},T)\right))\geq d.$$
 Since $k_0\neq 1$, $\deg_T(\mathcal{X}')=\deg_T (\mathcal{X})-1>\deg_T (\mathcal{X})-k_0=d$. By Theorem \ref{thm3} and the proof of Lemma \ref{cor1},
 $$\deg(\gcd \left(\mathcal{X}^{\prime}(\bar{u},T), \mathcal{X}(\bar{u},T)\right))\leq d \iff  \operatorname{PSC}_{d}(\mathcal{X},\mathcal{X}^\prime)(\bar{u})\neq 0.$$
 Let $N_1=N \times \{\operatorname{PSC}_{d}(\mathcal{X},\mathcal{X}^\prime)\}$. Then,
	$$\deg(\gcd \left(\mathcal{X}^{\prime}(\bar{u},T), \mathcal{X}(\bar{u},T)\right))= d \iff  \bar{u} \in \V_L(E)\backslash \V_L(N_1)\neq\emptyset.$$
 Thus, $t$ is separating element of $I(\bar{u})$ if and only if $\bar{u} \in \V_L(E)\backslash \V_L(N \times \{\operatorname{PSC}_{d}(\mathcal{X},\mathcal{X}^\prime)\})\neq\emptyset$.
\end{proof}

 \begin{remark}
 \begin{enumerate}
 \item For non-parametric cases, choosing $t$ from $S$ randomly, $t$ is a separating element of some zero-dimensional ideal with probability 1. Therefore, it is easy to choose a $t$ satisfying Theorem \ref{thm33}.
 \item For $E,N\in k[U]$, we can use the method in \cite{Kapur13} to check whether $\V_L(E)\backslash \V_L(N)$ $ = \emptyset$.
 \end{enumerate}
 \end{remark}

\begin{theorem}\label{thm44}
Let $I=\langle f_1(U,X),\ldots,f_l(U,X)\rangle \subset k[U,X]$, $E,N \subset k[U]$, $G \subset k[U,X]$. Assume that for all $\bar{u} \in \V_L(E)\backslash \V_L(N)$,
\begin{enumerate}
	\item $I(\bar{u})$ is zero-dimensional and $G(\bar{u})$ is a \gr basis of $I(\bar{u})$;
	\item for each $g\in G$, $c_1(\bar{u})\neq 0$, where $c_1$ is the leading coefficient of $g$ w.r.t. $X$;
	\item the number of points in $\V_L(I(\bar{u}))$ is equal to $k_0$.
\end{enumerate}
Then there is a finite set
\[\{(E_1,N_1,t_1),\ldots,(E_s,N_s,t_s)\},~E_i,N_i \subset k[U],\;t_i \in k[X],\]
such that $\V_L(E) \backslash \V_L(N)=\cup_{i=1}^s\V_L(E_i) \backslash \V_L(N_i)$ and for each $i \in \{1,\ldots,s\}$,  $t_i$ is a separating element of $I(\bar{u})$ for $\bar{u} \in \V_L(E_i)\backslash \V_L(N_i)$.
\end{theorem}

\begin{proof}
 If $k_0=1$, then $t_1=1$, $E_{1}=E$, $N_{1}=N$ and it is done. When $k_0\neq1$, we choose $t\in S=\{x_{1}+i x_{2}+\cdots+i^{n-1} x_{n}, i=1, \ldots, n k_0(k_0-1) / 2\}$. By Theorem \ref{thm33}, if $\V_L(E)\backslash \V_L(N \times \{\operatorname{PSC}_{d}(\mathcal{X},\mathcal{X}^\prime)\})$ $=\emptyset$, pick another $t$ from $S$; otherwise, let $t_1=t, E_1=E, N_1=N \times \{\operatorname{PSC}_{d}(\mathcal{X},\mathcal{X}^\prime)\}, E_{10}=E \cup \{\operatorname{PSC}_{d}(\mathcal{X},\mathcal{X}^\prime)\}, N_{10}=N$. Then
 $$\V_L(E)\backslash \V_L(N)=\V_L(E_{1}) \backslash \V_L(N_{1}) \cup \V_L(E_{10})\backslash \V_L(N_{10})$$
 and $t_1$ is a separating element of $I(\bar{u})$ for $\bar{u} \in \V_L(E_1)\backslash \V_L(N_1)$. If $\V_L(E_{10})\backslash \V_L(N_{10})=\emptyset$, it is done. Otherwise, it implies from $\V_L(E_{1}) \backslash \V_L(N_{1})\neq\emptyset$ that $\operatorname{PSC}_{d}(\mathcal{X},\mathcal{X}^\prime)\notin \langle E\rangle$. Thus $\langle E \rangle \subsetneq \langle E_{10} \rangle$.  Continually, using Theorem \ref{thm33} to $(E_{10}, N_{10})$ and choosing $t_2\in S$, we have
 $$\V_L(E_{10})\backslash \V_L(N_{10})=\V_L(E_2) \backslash \V_L(N_2) \cup \V_L(E_{20})\backslash \V_L(N_{20})$$
 and $t_2$ is a separating element of $I(\bar{u})$ for $\bar{u} \in \V_L(E_2)\backslash \V_L(N_2)$. If $\V_L(E_{20})\backslash \V_L(N_{20})=\emptyset$, it is done. Otherwise, $\langle E_{10} \rangle \subsetneq \langle E_{20} \rangle$. Go on like this, we get strictly ascending ideal chains:
 $$E \subsetneq E_{10} \subsetneq E_{20} \subsetneq E_{30} \subsetneq \cdots.$$
 According to the ascending chains condition, it must stop in finite steps. Thus, we obtain a finite set
$$\{(E_1,N_1,t_1),\ldots,(E_s,N_s,t_s)\},~E_i,N_i \subset k[U],\;t_i \in k[X],$$
such that $\V_L(E) \backslash \V_L(N)=\cup_{i=1}^s\V_L(E_i) \backslash \V_L(N_i)$ and for each $i \in \{1,\ldots,s\}$, $t_i$ is a separating element of $I(\bar{u})$ for $\bar{u} \in \V_L(E_i)\backslash \V_L(N_i)$.
\end{proof}

 \subsubsection{Partitions according to GCDs of $\mathcal{X}_t$ and $\mathcal{X}_t^{\prime}$}

 Suppose that $(E,N,t)$ is one of branches in Theorem \ref{thm44}. That is, $t$ is a separating element of $I(\bar{u})$ for any $\bar{u} \in \V_L(E)\backslash \V_L(N)$. We know if one obtain the RUR then you have to compute the GCD of $\mathcal{X}_t$ and $\mathcal{X}^\prime_t$. However, for different $\bar{u} \in \V_L(E)\backslash \V_L(N)$, the GCD of $\mathcal{X}_t(\bar{u},T)$ and $\mathcal{X}^\prime_t(\bar{u},T)$ may be different. In order to obtain the RUR with parameters, one needs to use parametric GCD algorithm \citep{Kapur18,Wang20,Wang2022}. $\V_L(E) \backslash \V_L(N)$ will be divided into finite branches. In each branch the GCD of $\mathcal{X}_t$ and $\mathcal{X}^\prime_t$ has the same expression, then by adjusting Formula (\ref{eq1}), (\ref{eq2}) in Section 2.1 to $k(U)[T]$, we can compute $g(U,T)$, $g_1(U,T)$, $\ldots$, $g_n(U,T)$ and obtain an RUR for each branch.

 \begin{theorem}\label{thm3.9}
 Let $I=\langle f_1(U,X),\ldots,f_l(U,X)\rangle \subset k[U,X]$, $E,N \subset k[U]$, $G \subset k[U,X]$, $t \in k[X]$. Assume that for all $\bar{u} \in \V_L(E)\backslash \V_L(N)$,
	\begin{enumerate}
		\item $I(\bar{u})$ is zero-dimensional and $G(\bar{u})$ is a \gr basis of $I(\bar{u})$;
		\item for each $g\in G$, $c_1(\bar{u})\neq 0$, where $c_1$ is the leading coefficient of $g$ w.r.t. $X$;
		\item $t$ is a separating element of $I(\bar{u})$.
	\end{enumerate}
 Let $\mathcal{X}_t(U,T)\in k(U)[T]$ be the characteristic polynomial as defined in Lemma \ref{lem3.4}. Then there is a finite set
$$\{(E_1,N_1,\mathcal{X}_t,g_1,g_{11},\ldots,g_{1n}),\ldots,(E_s,N_s,\mathcal{X}_t,g_s,g_{s1},\ldots,g_{sn})\},$$
such that $\V_L(E)\backslash \V_L(N)=\bigcup_{i=1}^{s}\V_L(E_i)\backslash \V_L(N_i)$ and for each $i \in \{1,\ldots,s\}$ and for each $\bar{u} \in \V_L(E_i)\backslash \V_L(N_i)$,
$$\V_L(I(\bar{u}))=\left\{\left(\frac{g_{i1}(\bar{u},\beta)}{g_i(\bar{u},\beta)},\ldots,\frac{g_{in}(\bar{u},\beta)}{g_i(\bar{u},T)}\right)\Big|\;\beta \in \V_L(\mathcal{X}_t(\bar{u},T))\right\}, $$
where $E_i, N_i \in k[U]$, $\mathcal{X}_t, g_{i},g_{i1},\ldots,g_{in} \in k(U)[T]$.
\end{theorem}

\begin{proof}
 Let $\mathcal{X}(U,T)\in k[U][T]$ be the product of $\mathcal{X}_t(U,T)$ and the least common multiple of all denominators of $\mathcal{X}_t(U,T)$. By the parametric GCD algorithm, we can obtain the parametric GCDs of $\mathcal{X}(U,T)$ and $\mathcal{X}^{\prime}(U,T)$ on $\V_L(E)\backslash \V_L(N)$ and assume that it is
	$$\left\{\left(E_{1}, N_{1}, d_{1}\right), \ldots,\left(E_{s}, N_{s}, d_{s}\right)\right\},$$
 where $\V_L(E) \backslash \V_L(N)=\cup_{i=1}^{s} \V_L(E_i)\backslash \V_L(N_i)$, $d_{1}, \ldots, d_{s}$ $\in k[U][T]$ and for every $i=1, \ldots, s$, $d_i(\bar{u},T)$ is a greatest common divisor of $\mathcal{X}(\bar{u},T)$ and $\mathcal{X}^{\prime}(\bar{u},T)$ for any $\bar{u} \in \V_L(E_i)\backslash \V_L(N_i)$. Moreover, the leading coefficient of $d_i$ w.r.t. $T$ is not zero under $\bar{u} \in \V_L(E_i)\backslash \V_L(N_i)$. Therefore, we can divide $d_i$ by the leading coefficient of $d_i$ and make $d_i \in k(U)[T]$ monic. For $\bar{u} \in \V_L(E_i) \backslash \V_L(N_i)$, we have $$\gcd(\mathcal{X}_t(\bar{u},T),\mathcal{X}_t^{\prime}(\bar{u},T))=\gcd(\mathcal{X}(\bar{u},T),\mathcal{X}^{\prime}(\bar{u},T))=d_i(\bar{u},T).$$
 For each $i=1,\ldots,s$, using divide algorithm in $k(U)[T]$, we get
 $$\mathcal{X}_t(U,T)=q_i(U,T)d_i(U,T)+r_i(U,T),\;\deg_T(r_i)<\deg_T(d_i),$$
 Since $d_i(\bar{u},T)$ divides $\mathcal{X}_t(\bar{u},T)$, for all $\bar{u} \in \V_L(E_i) \backslash \V_L(N_i)$, $r_i(U,T)$ $=0$. Therefore,
 $$\frac{\mathcal{X}_t(\bar{u},T)}{\gcd \left(\mathcal{X}_{t}^{\prime}(\bar{u},T), \mathcal{X}_{t}(\bar{u},T)\right)}  =q_i(\bar{u},T).$$
 Let $\overline{\mathcal{X}_t}(U,T)=q_i(U,T)=\sum_{j=0}^{d}a_jT^{d-j} \in k(U)[T]$. Compute
 \begin{align}
 g_i(U,T)=&\sum_{i=0}^{d-1}\sum_{j=0}^{d-i-1}\text{Tr}(M_{t^i}^I)a_jT^{d-i-j-1}, \label{eq3}\\
 g_{ij}(U,T)=&\sum_{i=0}^{d-1}\sum_{j=0}^{d-i-1}\text{Tr}(M_{x_kt^i}^I)a_jT^{d-i-j-1},j=1,\ldots,n. \label{eq4}
 \end{align}
 Then
 $$\{\mathcal{X}_t(\bar{u},T),g_i(\bar{u},T),g_{i1}(\bar{u},T),\ldots,g_{in}(\bar{u},T)\}$$
 is an RUR of $I(\bar{u})$ for $\bar{u} \in \V_L(E_i) \backslash \V_L(N_i)$.
\end{proof}

 We are now ready to give the first algorithm for computing the rational univariate representations with parameters of a ideal in $k[U,X]$, where we deliberately avoid tricks and optimizations, such as the selection of separating elements, reducing $\mathcal{X}$ or $g_i$ by $E_i$.
\begin{center}
\textbf{Algorithm 1 for computing RURs with parameters}
\end{center}

{\bf Input:} $I=\langle f_1,\ldots,f_l\rangle \subset k[U,X]$.

\begin{enumerate}
\item Compute a comprehensive \gr system of $I$ with $\prec_{X,U}$. Choose the zero-dimensional branches by Theorem \ref{thm3.2} and denote them by $\mathcal{G}=\{(E_1,N_1,G_1),\ldots,(E_s,N_s,G_{s})\}$.
\item $B_2:=\emptyset$. For each branch $(E_i,N_i,G_i) \in \mathcal{G}$ do:
\begin{enumerate}
	\item Compute $Q_1^I=(p_{ij}/q_{ij}) \in k(U)^{r \times r}$ defined in Lemma \ref{lem3.4}. Let $Q=D \cdot Q_1^I$, where $D=\text{lcm}(\{q_{ij}:i,j=1,\ldots,r\})$.
	\item Assume that $Q=(f_{ij})_{1 \leq i,j \leq r}$. Construct polynomial $f_i=\sum_{j=1}^rf_{ij}z_j$. Compute the minimal Comprehensive \gr system $\mathcal{G}=\{(E_{i1},N_{i1},G_{i1}),\ldots,(E_{is_i},N_{is_i},G_{is_i})\}$ of $\{f_1,\ldots,f_r\}$. Set
	$A_{i}:=\{(E_{i1},N_{i1},k_{i1}),$ $ \ldots$, $(E_{is_i},N_{is_i},k_{is_i})\}$, where $k_{ij}$ is equal to number of element in $G_{ij}$, $j=1,\ldots,s_i$.
	\item $B_2:=B_2\cup A_i$.
\end{enumerate}
\item $B_3:=\emptyset$. For each branch $(E_{ij},N_{ij},k_{ij}) \in B_2$ do:
\begin{enumerate}
	\item If $k_{ij}=1$ then $A_{ij}:=\{(E_{ij},N_{ij},1)\}$ and goto (3.d); else $A_{ij}:=\emptyset$, $S:=\{x_{1}+i x_{2}+\cdots+i^{n-1} x_{n}, i=1, \ldots, n k_{ij}(k_{ij}-1) / 2\}$.
	\item Randomly choose $t\in S$. Compute $M_t^I \in k(U)^{r\times r}$, $\mathcal{X}_t \in k(U)[T]$. Let $\mathcal{X}=D' \cdot \mathcal{X}_t$. Compute $\text{PSC}_d(\mathcal{X},\mathcal{X}')$, where $D'$ is the least common multiple of all denominators of $\mathcal{X}_t$ and $d=\deg_T(\mathcal{X})-k_{ij}$.
	\item If $\V_L(E_{ij})\backslash \V_L(N_{ij} \times \text{PSC}_d(\mathcal{X},\mathcal{X}'))=\emptyset$, then $S:=S\backslash \{t\}$ and goto (3.b); else, $A_{ij}:=A_{ij} \cup \{(E_{ij},N_{ij} \times \text{PSC}_d(\mathcal{X},\mathcal{X}'),t)\}$. Update $E_{ij}:=E_{ij} \cup \{\text{PSC}_d(\mathcal{X},\mathcal{X}')\}$ and $S:=S\backslash \{t\}$. If $\V_L(E_{ij})\backslash \V_L(N_{ij}) \neq \emptyset$, then goto (3.b).
	\item $B_3:=B_3 \cup A_{ij}$.
\end{enumerate}
\item $B_4:=\emptyset$. For each branch $(E_{ijk},N_{ijk},t_{ijk})\in B_3$ do:
\begin{enumerate}
	\item  Compute parametric GCDs of $\mathcal{X}$ and $\mathcal{X}'$, denoted by
          $$\{(E_{ijkl},N_{ijkl},d_{ijkl}):l=1,\ldots,v\},$$
          where $\mathcal{X}_t \in k(U)[T]$ and $\mathcal{X} \in k[U][T]$ are the polynomials computed in (3.b) corresponding to $t_{ijk}$.
	\item For each $(E_{ijkl},N_{ijkl},d_{ijkl})$ do:

  $A_{ijkl}:=\emptyset$. Compute $\overline{\mathcal{X}_{t}}$. By Formula (3), (4) compute $g_{ijkl}(U,T)$, $g_{ijkl1}(U,T)$, \ldots, $g_{ijkln}(U,T)$. $A_{ijkl}:=A_{ijkl}\cup \{(E_{ijkl},N_{ijkl},\mathcal{X}_t,g_{ijkl},g_{ijkl1},\ldots,g_{ijkln})\}$
   \item $B_4:=B_4\cup A_{ijkl}$.
\end{enumerate}
\end{enumerate}

{\bf Output:} a finite set $B_4$ of the rational univariate representations with parameters, renumbering it by
 $$\left\{(E_1,N_1,\mathcal{X}_1,g_1,g_{11},\ldots,g_{1n}),\ldots,(E_s,N_s,\mathcal{X}_s,g_s,g_{s1},\ldots,g_{sn})\right\}.$$

 For each $i \in \{1,\ldots,s\}$, it satisfies that :
 for all $\bar{u} \in \V_L(E_i)\backslash \V_L(N_i)$,  $$\V_L(I(\bar{u}))=\left\{\left(\frac{g_{i1}(\bar{u},\beta)}{g_i(\bar{u},\beta)},\ldots,\frac{g_{in}(\bar{u},\beta)}{g_i(\bar{u},\beta)}\right)\Big|\;\beta \in \V_L(\mathcal{X}_i(\bar{u},T))\right\},$$
 where $I(\bar{u})=\langle f_1(\bar{u},X),\ldots,f_l(\bar{u},X) \rangle \subset k[X]$, $\mathcal{X}_i(\bar{u},T) \in k[T].$

 \begin{theorem}
 The above algorithm works correctly and terminates in a finite number of steps.
 \end{theorem}

 \begin{proof}
 The correctness and the termination of the algorithm directly follows from Theorems \ref{thm3.2}, \ref{thm22}, \ref{thm44} and \ref{thm3.9}.
 \end{proof}

\subsection{Algorithm 2}

 Based on three steps of the second idea, we can design another algorithm different from Algorithm 1. Since the first two steps are same, we will only elaborate on Step 3$^\spadesuit$.

 Assume that for all $\bar{u} \in \V_L(E) \backslash \V_L(N)$, the number of points in $\V_L(I(\bar{u}))$ is the same, denoted by $k_0$.
 Choose randomly a $t \in \{x_{1}+i x_{2}+\cdots+i^{n-1} x_{n}, i=1,\ldots, n k_0(k_0-1) / 2\}$ and compute the characteristic polynomial $\mathcal{X}_t(U,T)\in k(U)[T]$ as defined in Lemma \ref{lem3.4}. Let $\mathcal{X}(U,T)\in k[U][T]$ be the product of $\mathcal{X}_t(U,T)$ and the least common multiple of all denominators of $\mathcal{X}_t(U,T)$. By the parametric GCD algorithm, we can obtain the parametric GCDs of $\mathcal{X}(U,T)$ and $\mathcal{X}^{\prime}(U,T)$ on $\V_L(E)\backslash \V_L(N)$ and assume that it is
	$$\left\{\left(E_{1}, N_{1}, d_{1}\right), \ldots,\left(E_{r}, N_{r}, d_{r}\right)\right\},$$
 where $\V_L(E) \backslash \V_L(N)=\cup_{i=1}^{r} \V_L(E_i)\backslash \V_L(N_i)$, $d_{1}, \ldots, d_{r}$ $\in k[U][T]$ and for every $i=1, \ldots, r$, $d_i(\bar{u},T)$ is a greatest common divisor of $\mathcal{X}(\bar{u},T)$ and $\mathcal{X}^{\prime}(\bar{u},T)$ for any $\bar{u} \in \V_L(E_i)\backslash \V_L(N_i)$. Moreover, the leading coefficient of $d_i$ w.r.t. $T$ is not zero under $\bar{u} \in \V_L(E_i)\backslash \V_L(N_i)$.

 As stated in the proof of Theorem \ref{thm33}, for $\bar{u} \in \V_L(E) \backslash \V_L(N)$, $t$ is a separating element of $I(\bar{u})$ if and only
 $$\deg(\gcd \left(\mathcal{X}^{\prime}(\bar{u},T), \mathcal{X}(\bar{u},T)\right))=\deg_T (\mathcal{X}(U,T)) -k_0.$$
 Since $$\deg(\gcd \left(\mathcal{X}^{\prime}(\bar{u},T), \mathcal{X}(\bar{u},T)\right))\geq \deg_T (\mathcal{X}(U,T)) -k_0,$$
 combining with Remark 16, it implies that there exists a unique $i\in \{1,\ldots,r\}$ such that $t$ is a separating element of $I(\bar{u})$ for $\bar{u} \in \V_L(E_i)\backslash \V_L(N_i)$.

 Then we obtain the following theorem.

 \begin{theorem}\label{thm55}
 Let $I=\langle f_1(U,X),\ldots,f_l(U,X)\rangle \subset k[U,X]$, $E,N \subset k[U]$, $G \subset k[U,X]$. Assume that for all $\bar{u} \in \V_L(E)\backslash \V_L(N)$,
 \begin{enumerate}
	\item $I(\bar{u})$ is zero-dimensional and $G(\bar{u})$ is a \gr basis of $I(\bar{u})$;
	\item for each $g\in G$, $c_1(\bar{u})\neq 0$, where $c_1$ is the leading coefficient of $g$ w.r.t. $X$;
	\item the number of points in $\V_L(I(\bar{u}))$ is equal to $k_0$.
 \end{enumerate}
 Then there is a finite set
 \[\{(E_1,N_1,t_1,d_1),\ldots,(E_s,N_s,t_s,d_s)\},~E_i,N_i \subset k[U],\;t_i \in k[X],d_i \in k[X][T],\]
 such that $\V_L(E) \backslash \V_L(N)=\cup_{i=1}^s\V_L(E_i) \backslash \V_L(N_i)$ and for each $i \in \{1,\ldots,s\}$,  $t_i$ is a separating element of $I(\bar{u})$ and $d_i(\bar{u},T)$ is a greatest common divisor of $\mathcal{X}(\bar{u},T)$ and $\mathcal{X}^{\prime}(\bar{u},T)$ for any $\bar{u} \in \V_L(E_i)\backslash \V_L(N_i)$.
 \end{theorem}

 \begin{proof}
 We choose $t\in S=\{x_{1}+i x_{2}+\cdots+i^{n-1} x_{n}, i=1, \ldots, n k_0(k_0-1) / 2\}$. As mentioned above, we can compute the parametric GCDs of $\mathcal{X}(U,T)$ and $\mathcal{X}^{\prime}(U,T)$ on $\V_L(E)\backslash \V_L(N)$ and assume that it is
	$$\left\{\left(\tilde{E}_{1}, \tilde{N}_{1}, \tilde{d}_{1}\right), \ldots,\left(\tilde{E}_{r}, \tilde{N}_{r}, \tilde{d}_{r}\right)\right\}.$$
 If $\deg(\tilde{d}_i)> \deg_T (\mathcal{X}(U,T)) -k_0$  for every $i\in \{1,\ldots,r\}$, pick another $t$ from $S$; otherwise, there exists a unique $i\in \{1,\ldots,r\}$ such that $\deg(\tilde{d}_i)= \deg_T (\mathcal{X}(U,T)) -k_0$, then $t$ is a separating element of $I(\bar{u})$ for $\bar{u} \in \V_L(\tilde{E}_i)\backslash \V_L(\tilde{N}_i)$.  Without loss of generality, we assume $i=1$, let $E_1=\tilde{E}_1, N_1=\tilde{N}_1, t_1=t, d_1=\tilde{d}_1$. If $r=1$, it is done. Otherwise, it implies from Theorem \ref{thm33} that $\operatorname{PSC}_{\deg_T (\mathcal{X}(U,T)) -k_0}(\mathcal{X},\mathcal{X}^\prime)\in \sqrt{\langle \tilde{E}_{i}\rangle}$ for $i=2, \ldots, r$. Since $\V_L(E_1) \backslash \V_L(N_1)\neq\emptyset$, $\operatorname{PSC}_{\deg_T (\mathcal{X}(U,T)) -k_0}(\mathcal{X},\mathcal{X}^\prime)\notin \sqrt{\langle E\rangle}$. Thus $\sqrt{\langle E \rangle} \subsetneq \sqrt{\langle \tilde{E}_{i}\rangle}$ for $i=2, \ldots, r$. Now, for the remaining $r-1$ branches, choosing $t_2\in S$ and computing the parametric GCDs of $\mathcal{X}(U,T)$ and $\mathcal{X}^{\prime}(U,T)$ on these branches, we can obtain $(E_2,N_2,t_2,d_2)$. Go on like this, we get strictly ascending ideal chains:
 $$E \subsetneq \tilde{E}_{i} \subsetneq \tilde{E}_{ij} \subsetneq \tilde{E}_{ijk} \subsetneq \cdots,$$
 where $i,j,k\neq1$. According to the ascending chains condition, it must stop in finite steps. Thus, we obtain a finite set
\[\{(E_1,N_1,t_1,d_1),\ldots,(E_s,N_s,t_s,d_s)\},~E_i,N_i \subset k[U],\;t_i \in k[X],d_i \in k[X][T],\]
 such that $\V_L(E) \backslash \V_L(N)=\cup_{i=1}^s\V_L(E_i) \backslash \V_L(N_i)$ and for each $i \in \{1,\ldots,s\}$,  $t_i$ is a separating element of $I(\bar{u})$ and $d_i(\bar{u},T)$ is a greatest common divisor of $\mathcal{X}(\bar{u},T)$ and $\mathcal{X}^{\prime}(\bar{u},T)$ for any $\bar{u} \in \V_L(E_i)\backslash \V_L(N_i)$.
\end{proof}

 According to Theorem \ref{thm3.9}, we further compute $g(U,T)$, $g_1(U,T)$, $\ldots$, $g_n(U,T)$ and obtain an RUR for each branch.

 Now we design another algorithm for computing the rational univariate representations with parameters of a ideal in $k[U,X]$.
\begin{center}
\textbf{Algorithm 2 for computing RURs with parameters}
\end{center}

{\bf Input:} $I=\langle f_1,\ldots,f_l\rangle \subset k[U,X]$.

\begin{enumerate}
\item Compute a comprehensive \gr system of $I$ with $\prec_{X,U}$. Choose the zero-dimensional branches by Theorem \ref{thm3.2} and denote them by $\mathcal{G}=\{(E_1,N_1,G_1),\ldots,(E_s,N_s,G_{s})\}$.
\item $B_2:=\emptyset$. For each branch $(E_i,N_i,G_i)\in\mathcal{G}$ do:
\begin{enumerate}
	\item Compute $Q_1^I=(p_{ij}/q_{ij}) \in k(U)^{r \times r}$ defined in Lemma \ref{lem3.4}.
          Let $Q=D \cdot Q_1^I$, where $D=\text{lcm}(\{q_{ij}:i,j=1,\ldots,r\})$.
	\item Assume that $Q=(f_{ij})_{1 \leq i,j \leq r}$. Construct polynomial $f_i=\sum_{j=1}^rf_{ij}z_j$.
          Compute the minimal Comprehensive \gr system $\mathcal{G}'=\{(E_{i1},N_{i1},G_{i1}),\ldots,(E_{is_i},N_{is_i},G_{is_i})\}$ of $\{f_1,\ldots,f_r\}$.
          Set $A_{i}:=\{(E_{i1},N_{i1},k_{i1}),$ $ \ldots$, $(E_{is_i},N_{is_i},k_{is_i})\}$, where $k_{ij}$ is equal to number of element in $G_{ij}$, $j=1,\ldots,s_i$.
	\item $B_2:=B_2\cup A_i$.
\end{enumerate}
\item $B_3:=\emptyset$. For each branch $(E_{ij},N_{ij},k_{ij}) \in B_2$ do:
\begin{enumerate}
	\item Randomly choose $t\in S$. Compute $M_t^I \in k(U)^{r\times r}$, $\mathcal{X}_t \in k(U)[T]$.
          Let $\mathcal{X}=D' \cdot \mathcal{X}_t$. Compute parametric GCDs of $\mathcal{X}$ and $\mathcal{X}'$, denoted by
          $$\{(E_{ijk},N_{ijk},d_{ijk}):k=1,\ldots,v\},$$
          where $D'$ is the least common multiple of all denominators of $\mathcal{X}_t$.
	\item For each $(E_{ijk},N_{ijk},d_{ijk})$ do:

          $A_{ijk}:=\emptyset$. If $\deg(d_{ijk})> \deg_T (\mathcal{X}(U,T)) -k_{ij}$ , then $S:=S\backslash \{t\}$ and goto (3.a); else, compute $\overline{\mathcal{X}_{t}}$. By Formula (3), (4) compute $g_{ijk}(U,T)$, $g_{ijk1}(U,T)$, \ldots, $g_{ijkn}(U,T)$. $A_{ijk}:=A_{ijk}\cup \{(E_{ijk},N_{ijk},\mathcal{X}_t,g_{ijk},g_{ijk1},\ldots,g_{ijkn})\}$
	\item $B_3:=B_3 \cup A_{ijk}$.
\end{enumerate}
\end{enumerate}

{\bf Output:} a finite set $B_3$ of the rational univariate representations with parameters, renumbering it by
 $$\left\{(E_1,N_1,\mathcal{X}_1,g_1,g_{11},\ldots,g_{1n}),\ldots,(E_s,N_s,\mathcal{X}_s,g_s,g_{s1},\ldots,g_{sn})\right\}.$$

 For each $i \in \{1,\ldots,s\}$, it satisfies that :
 for all $\bar{u} \in \V_L(E_i)\backslash \V_L(N_i)$,  $$\V_L(I(\bar{u}))=\left\{\left(\frac{g_{i1}(\bar{u},\beta)}{g_i(\bar{u},\beta)},\ldots,\frac{g_{in}(\bar{u},\beta)}{g_i(\bar{u},\beta)}\right)\Big|\;\beta \in \V_L(\mathcal{X}_i(\bar{u},T))\right\},$$
 where $I(\bar{u})=\langle f_1(\bar{u},X),\ldots,f_l(\bar{u},X) \rangle \subset k[X]$, $\mathcal{X}_i(\bar{u},T) \in k[T].$

 \begin{theorem}
 The above algorithm works correctly and terminates in a finite number of steps.
 \end{theorem}

 \begin{proof}
 The correctness and the termination of the algorithm directly follows from Theorems \ref{thm3.2}, \ref{thm22}, \ref{thm55}.
 \end{proof}

\section{Example}\label{sec4}

 We use the following simple example to illustrate the steps in the above proposed algorithms.

 \begin{example}\label{ex}
	Let $I=\langle u_1x_1^2+u_2x_2+u_2,u_2x_2^2+u_1x_2+u_1 \rangle \subset \mathbb{C}[u_1,u_2,x_1,x_2]$ with $\prec_U$ and $\prec_X$ being graded lexicographic order. We see the process of Algorithm 1:

    (1) Compute the minimal comprehensive \gr system for $I$:
	$$\begin{aligned}
		\mathcal{G}=&\{(\{0\},\{u_1u_2\},\{u_2x_2^2+u_1x_2+u_1, u_1x_1^2+u_2x_2+u_2\}),\\
		&(\{u_1\},\{u_2, u_2u_1\},\{1\}),(\{u_2\},\{u_1\},\{u_1x_1^2, u_1x_2+u_1\}),\\
		&(\{u_1,u_2\},\{1\},\{0\})\}.
	\end{aligned}$$
    By Finiteness Theorem, $(\{0\},\{u_1u_2\},\{u_2x_2^2+u_1x_2+u_1, u_1x_1^2+u_2x_2+u_2\})$ and $(\{u_2\},\{u_1\},\{u_1x_1^2, u_1x_2+u_1\})$ are two zero-dimensional branches. For briefness, we only consider the first branch.

     (2) For $(\{0\},\{u_1u_2\},\{u_2x_2^2+u_1x_2+u_1, u_1x_1^2+u_2x_2+u_2\})$, we compute $Q^I_1(U)$ as defined in Theorem \ref{thm22}:
		{\scriptsize
        $$Q^I_1(U)=\left[\begin{matrix}
		4 & -\frac{2 u_{1}}{u_{2}} & 0 & 0 \\
		-\frac{2 u_{1}}{u_{2}} & \frac{2 u_{1}\left(u_{1}-2 u_{2}\right)}{u_{2}^{2}} & 0 & 0 \\
		0 & 0 & \frac{2\left(u_{1}-2 u_{2}\right)}{u_{1}} & -\frac{2\left(-3 u_{2}+u_{1}\right)}{u_{2}} \\
		0 & 0 & \frac{2\left(3 u_{2}-u_{1}\right)}{u_{2}} & \frac{2\left(u_{1}^{2}-4 u_{1} u_{2}+2 u_{2}^{2}\right)}{u_{2}^{2}}
	\end{matrix}\right].$$
}Then $D=u_1u_2^2$ be the least common multiple of all denominators of $Q_1^I$. Let $Q=D\cdot Q^I_1$. Denote the entries of $Q$ by $f_{ij}$ and set $f_i=\sum_{j=1}^{4}f_{ij}z_j$, $i=1,\ldots,4$. We compute the minimal comprehensive \gr system of $\{f_1,f_2,f_3,f_4\}$ on $\mathbb{C}^2\setminus\V_{\mathbb{C}}(u_1u_2)$: $\{\big(\{0\}, \{(u_1-4u_2)u_2u_1\}, \{(u_1^3u_2-4u_1^2u_2^2)z_2, (u_1^2u_2^3-4u_1u_2^4)z_4, 2u_1u_2^2z_1-u_1^2u_2z_2,$ $2u_2^4z_3-u_1u_2^3z_4\}\big), \big(\{u_1-4u_2\}, \{u_1u_2\}, $ $\{u_2^3z_1-2u_2^3z_2, u_2^3z_3-2u_2^3z_4\}\big)\}.$
	Therefore,
	\begin{align*}
	B_2=&\{(E_{11},N_{11},k_{11}),(E_{12},N_{12},k_{12})\}\\
       =&\{(\{0\},\{u_1u_2(u_1-4u_2)\},4),(\{u_1-4u_2\},\{u_1u_2\},2)\}.
	\end{align*}

    (3) Let $B_3=\emptyset$. For the branch $(E_{11},N_{11},k_{11})=(\{0\},\{u_1u_2(u_1-4u_2)\},4)$, let $A_{11}=\emptyset$. Randomly choose $t=x_1\in S$. By computation,
    $$\begin{aligned}
    M_t^I&=\left[\begin{smallmatrix}
    	0 & 0 & 1 & 0
    	\\
    	0 & 0 & 0 & 1
    	\\
    	-\frac{u_{2}}{u_{1}} & -\frac{u_{2}}{u_{1}} & 0 & 0
    	\\
    	1 & \frac{u_{1}-u_{2}}{u_{1}} & 0 & 0
    \end{smallmatrix}\right],\quad \mathcal{X}_t=T^{4}-\frac{(u_{1}-2 u_{2}) T^{2}}{u_{1}}+\frac{u_{2}^{2}}{u_{1}^{2}}.
    \end{aligned}.$$
     Then $\mathcal{X}=u_1^2T^{4}-u_1(u_{1}-2 u_{2}) T^{2}+u_{2}^{2}$, $\mathcal{X}'=4u_1^2T^{3}-2u_1(u_{1}-2 u_{2}) T$, $\text{PSC}_0(\mathcal{X},\mathcal{X}')=16 u_{1}^{10} u_{2}^{2} (u_{1}-4 u_{2})^{2}$.
    By check, $\mathbb{C}^2\backslash \V_{\mathbb{C}}(u_1u_2(u_1-4u_2)\cdot \text{PSC}_0(\mathcal{X},\mathcal{X}'))\neq \emptyset$, $x_1$ is a separating element of $I(\bar{u})$ under $\bar{u}\in\mathbb{C}^2\backslash \V_{\mathbb{C}}(16u_1^{11}u_2^3(u_1-4u_2)^3)$. Now
    $$A_{11}:=A_{11}\cup \{(\{0\},\{16u_1^{11}u_2^3(u_1-4u_2)^3\},x_1)\}.$$
    Update $E_{11}=E_{11}\cup \{\text{PSC}_0(\mathcal{X},\mathcal{X}')\}$. We can check $\V_{\mathbb{C}}(E_{11})\backslash \V_{\mathbb{C}}(N_{11})$ $=\emptyset$. Then this branch has been done and
    $$B_3:=B_3 \cup A_{11}=\{(\{0\},\{16u_1^{11}u_2^3(u_1-4u_2)^3\},x_1)\}.$$
	
	For the branch $(E_{12},N_{12},k_{12})=(\{u_1-4u_2\},\{u_1u_2\},2)$, let $A_{12}=\emptyset$. We also choose $t=x_1$. Then $M_t^I$, $\mathcal{X}_t$, $\mathcal{X}$, $\mathcal{X}'$ are the same as above. By computation,
	$$\text{PSC}_2(\mathcal{X},\mathcal{X}')=-8 u_{1}^{5} (u_{1}-2 u_{2}).$$
	We can check that $\V_{\mathbb{C}}(u_1-4u_2)\backslash \V_{\mathbb{C}}(u_1u_2\cdot \text{PSC}_2(\mathcal{X},\mathcal{X}'))\neq \emptyset$, then
    $$A_{12}:=A_{12}\cup \{(\{u_1-4u_2\},\{-8u_1^6u_2(u_1-2u_2)\},x_1)\}.$$
	Update $E_{12}=E_{12}\cup \{\text{PSC}_2(\mathcal{X},\mathcal{X}')\}=\{u_1-4u_2,-8u_1^5(u_1-2u_2)\}$. By check, $\V_{\mathbb{C}}(E_{12})\backslash \V_{\mathbb{C}}(N_{12})=\emptyset$. This branch has been done and
	\begin{align*}
		B_3:=&B_3 \cup A_{12}\\
            =&\{(\{0\},\{u_1^{11}u_2^3(u_1-4u_2)^3\},x_1),(\{u_1-4u_2\},\{-8u_1^{6}u_2(u_1-2u_2)\},x_1)\}.
	\end{align*}

	(4) For the branch $(\{0\},\{u_1^{11}u_2^3(u_1-4u_2)^3\},x_1)$, we compute the parametric GCDs of $\mathcal{X}$ and $\mathcal{X}'$ corresponding to $t=x_1$.
 That is $\{(\{0\},{\{u_{2} u_{1} (u_{1}-4 u_{2})\}},1)\}$.
 Then $\overline{\mathcal{X}_t}=\mathcal{X}_t$. By Formula (3),(4),  $g_1=4 T^{3}-\frac{2(u_{1}-2 u_{2}) T}{u_{1}},$
	$g_{11}=\frac{2\left(u_{1}-2 u_{2}\right)}{u_{1}}T^{2}-\frac{4 u_{2}^{2}}{u_{1}^{2}},$
	$g_{12}=-\frac{2 u_{1} }{u_{2}}T^{3}+2 T.$
	Let $\mathcal{X}_1=\mathcal{X}_t=T^{4}-\frac{u_{1}-2 u_{2} }{u_{1}}T^{2}+\frac{u_{2}^{2}}{u_{1}^{2}}$, we obtain the RUR with parameters of this branch is $(\{0\},\{u_1u_2(u_1-4u_2)\},\mathcal{X}_1,g_1,g_{11},g_{12}).$
	
	For the branch $(\{u_1-4u_2\},\{-8u_1^6u_2(u_1-2u_2)\},x_1)$, we compute the parametric GCDs of $\mathcal{X}$, $\mathcal{X}'$ corresponding to $t=x_1$.
 That is
		$$({\{u_{1}-4 u_{2}\}},{\{u_1u_2\}},4 T^{2} u_{2}^{2}-u_{2}^{2})$$
	By computation, $\overline{\mathcal{X}_{t}}=T^{2}-\frac{3 u_{1}-8 u_{2}}{4 u_{1}}$, $g_2=4T$, $g_{21}=\frac{2 (u_{1}-2 u_{2})}{u_{1}}$, $g_{22}=-\frac{2 u_{1} T}{u_{2}}$. Let $\mathcal{X}_2=\mathcal{X}_t$. In this case, we can reduce $\mathcal{X}_2$, $g_2$, $g_{21}$, $g_{22}$ by the relation $u_1-4u_2=0$. Then we have
	$\mathcal{X}_2=T^4-\frac{1}{2}T^2+\frac{1}{16}$, $g_2=4T$, $g_{21}=1,g_{22}=-8T$.
	Thus, the RUR with parameters of this branch is $({\{u_{1}-4 u_{2}\}},{\{(u_{1}-2 u_{2}) u_{1} u_{2}\}},\mathcal{X}_2,g_2,g_{21},g_{22}).$
	
	
	In conclusion, we obtain two RURs of $I$ under the parameter branch $(\{0\},\{u_1u_2\})$. That is $(\{0\},\{u_1u_2(u_1-4u_2)\},\mathcal{X}_1,g_1,g_{11},g_{12})$ and $({\{u_{1}-4 u_{2}\}},{\{(u_{1}-2 u_{2}) u_{1} u_{2}\}},\mathcal{X}_2,g_2,g_{21},g_{22})$. Therefore,
	
	(i) For all $\bar{u}=(\bar{u}_1,\bar{u}_2) \in \mathbb{C}^2 \backslash \V_{\mathbb{C}}(u_1u_2(u_1-4u_2))$,
	\begin{align*}
		\V_{\mathbb{C}}(I(\bar{u}))
		=&\left\{\left(\frac{g_{11}(\bar{u},\beta)}{g_1(\bar{u},\beta)},\frac{g_{12}(\bar{u},\beta)}{g_{1}(\bar{u},\beta)}\right)\Big|\;\beta \in \V_{\mathbb{C}}\left(T^{4}-\frac{\bar{u}_{1}-2 \bar{u}_{2} }{\bar{u}_{1}}T^{2}+\frac{\bar{u}_{2}^{2}}{\bar{u}_{1}^{2}}\right)\right\}\\
		=&\left\{\left(\frac{(\bar{u}_{1}^{2}-2 \bar{u}_{1} \bar{u}_{2}) \beta^{2}-2 \bar{u}_{2}^{2}}{2 \beta^{3} \bar{u}_{1}^{2}-(\bar{u}_{1}-2 \bar{u}_{2}) \bar{u}_{1} \beta},\frac{-\bar{u}_{1}^{2}\beta^{2}+\bar{u}_{1} \bar{u}_{2}}{2 \beta^{2} \bar{u}_{1} \bar{u}_{2}-\bar{u}_{2}\left(\bar{u}_{1}-2 \bar{u}_{2}\right)}\right)\Big|\right. \\
		&\left.\beta \in \V_{\mathbb{C}}\left(T^{4}-\frac{\bar{u}_{1}-2 \bar{u}_{2} }{\bar{u}_{1}}T^{2}+\frac{\bar{u}_{2}^{2}}{\bar{u}_{1}^{2}}\right)\right\}.
	\end{align*}
	
	(ii) For all $\bar{u}=(\bar{u}_1,\bar{u}_2) \in \V_{\mathbb{C}}(u_1-4u_2)\backslash \V_{\mathbb{C}}(u_2)$,
	$$\V_{\mathbb{C}}(I(\bar{u}))=\left\{\left(\frac{1}{4\beta},-2\right)\Big|\;\beta \in \V_{\mathbb{C}}\left(T^4-\frac{1}{2}T^2+\frac{1}{16}\right)\right\},$$
	
\end{example}
\begin{example}\label{ex2}
	Let $I$ be the same as above. We see the process of Algorithm 2. For briefness, we also consider the first branch:
	$$(\{0\},\{u_1u_2\},\{u_2x_2^2+u_1x_2+u_1, u_1x_1^2+u_2x_2+u_2\}).$$
	Step 2 is the same as Algorithm 1, then we get
	\begin{align*}
		B_2=&\{(E_{11},N_{11},k_{11}),(E_{12},N_{12},k_{12})\}\\
		=&\{(\{0\},\{u_1u_2(u_1-4u_2)\},4),(\{u_1-4u_2\},\{u_1u_2\},2)\}.
	\end{align*}

For branch $(\{0\},\{u_1u_2(u_1-4u_2)\},4)$. We choose $t = x_1$. Compute $\mathcal{X} = u_1^2T^4-u_1(u_1-2u_2)T^2+u_2^2$ and $\mathcal{X}' = 4u_1^2T^3-2u_1(u_1-2u_2)T$. Then we compute the parametric GCDs of $\mathcal{X}$ and $\mathcal{X}'$ over $\mathbb{C}^2\setminus\V_L(u_1u_2(u_1-4u_2))$ immediately. That is
$$\{(\{0\},\{u_1u_2(u_1-4u_2),1\})\}.$$
By the computation of parametric GCDs, we know that, in fact, $t = x_1$ is separate element for all $I(\bar{u})$ when $\bar{u} \in \mathbb{C}^2\setminus\V_L(u_1u_2(u_1-4u_2))$. Then the computations of $g_1$, $g_{11}$ and $g_{12}$ are the same as Example \ref{ex}. We can get $g_1=4 T^{3}-\frac{2(u_{1}-2 u_{2}) T}{u_{1}},$
$g_{11}=\frac{2\left(u_{1}-2 u_{2}\right)}{u_{1}}T^{2}-\frac{4 u_{2}^{2}}{u_{1}^{2}},$
$g_{12}=-\frac{2 u_{1} }{u_{2}}T^{3}+2 T.$

For branch $(\{u_1-4u_2\},\{u_1u_2\},4)$. We choose $t = x_1$. Then $\mathcal{X}$ and $\mathcal{X}'$ are the same as above. We compute the parametric GCDs of $\mathcal{X}$ and $\mathcal{X}'$ over $\V_L(u_1-4u_2)\setminus\V_L(u_1u_2)$ immediately. That is
$$\{({\{u_{1}-4 u_{2}\}},{\{u_1u_2\}},4 T^{2} u_{2}^{2}-u_{2}^{2})\}.$$
Then the computations of $g_2$, $g_{21}$ and $g_{22}$ are the same as Example \ref{ex}. We can get $g_2=4T$, $g_{21}=1,g_{22}=-8T$. Finally, the result is the same as Example \ref{ex}.
\end{example}
\begin{remark}
In Algorithm 2, we skip the computation of $\operatorname{PSC}_d(\mathcal{X},\mathcal{X}')$, so it seems to be more effective than Algorithm 1. But in fact, by the computation of $\operatorname{PSC}_d(\mathcal{X},\mathcal{X}')$, we can divide the parameter space into smaller one, and this may accelerate the computation of parametric GCD. Therefore, it is hard to tell which algorithm is better. 	
\end{remark}
\section{Implementation and Comparative Performance}

 The algorithms for computing RURs of zero-dimensional ideals with parameters have been implemented on the computer algebra system {\em Singular}. The codes and examples are available on the web: \url{http://www.mmrc.iss.ac.cn/~dwang/software.html}.

 We try all examples in \cite{Montes06} and \cite{Nabeshima07}, addition with some other examples from \cite{Jerome07}, \cite{Robbiano13} and \cite{Ayad12}. Specifically, all examples are listed in appendix.

 For all these examples, we compute the \gr basis of ideals under the block order $\prec_{U,X}$ with $U \ll X$; within each block, $\prec_U$ and $\prec_X$ are graded reverse lexicographic orders.
 The following tables show the comparison of our two algorithms in the number of branches and running time in seconds. Timings were obtained on Intel(R) Core(TM) i5-8250U CPU @ 1.60GHz 1.80 GHz with 8GB Memory running Windows 10.
	
\begin{table}[h]
	\begin{tabular}{ccccc}
		\hline
		\multirow{2}{*}{Examples}& {\small Number of 0-dim} &{\small Number of 0-dim}  &  \multirow{2}{*}{Time of Alg. 1} & \multirow{2}{*}{Time of Alg. 2} \\
		~&  {\small branches in Alg.1} &{\small branches in Alg.2 }&  ~ & ~ \\
		\hline
		$F_1$& 0 & 0 & 0.048 & 0.040   \\
		$F_2$& 4 & 4 & 0.181 & 0.158   \\
		$F_3$& 0 & 0 & 0.090 &  0.089  \\
		$F_4$& 14 & 14 & 0.718 & 0.680 \\
		$F_5$& 4 & 4 & 0.255 & 0.196  \\
		$F_6$& 9 & 9 & 0.442 & 0.433   \\
		$F_7$& 10 & - & 36.693 &  $>$1h  \\
		$F_8$& 31 & 32 & 56.298 & 55.797  \\
		$S_1$& 3 & 3 & 0.105 &  0.096 \\
		$S_2$& 1 & 1 & 0.053 &  0.048  \\
		$S_3$& 0 & 0 & 0.019 &  0.019  \\
		$S_4$& 1 & 1 & 0.045 & 0.040   \\
		$S_5$& 43 & 39 & 2.768 & 2.392   \\
		$S_6$& 4 & 4 & 0.259 & 0.215   \\
		$S_7$& 8 & 8 & 2.064 & 0.994   \\
		$S_8$& 1 & 1 & 0.068 & 0.061   \\
		$S_9$& 3 & 3 & 0.223 &  0.243  \\
		$S_{10}$& 13 & 17 & 0.625 & 0.779   \\
		$S_{12}$& 13 & 14 & 0.798 &  0.909  \\
		$S_{13}$& 4 & 4 & 0.205 & 0.204   \\
		$S_{14}$& 0 & 0 & 0.076 & 0.075   \\
		$S_{15}$& 15 & 15 & 1.044 & 1.105   \\
		$S_{16}$& 4 & 4 & 0.433 & 0.426   \\
		$C_{1}$& 30 & 30 & 3.382 & 2.646   \\
		$R_{1}$& 7 & 7 & 0.233 & 0.248   \\
		$E_{1}$& 2 & 2 & 0.075 & 0.062   \\
		$E_{2}$& 2 & 2 & 0.277 & 0.291   \\
		$E_{3}$& 4 & 4 & 0.171 & 0.155   \\
		\hline
	\end{tabular}
\end{table}
 Some remarks about the implementation issues are listed in the following:
 \begin{itemize}
 	\item To speed up, in Algorithm 1, Step 3.(c), we first factorize $\operatorname{PSC}_d(\mathcal{X}, \mathcal{X}')$. Suppose that $\operatorname{PSC}_d(\mathcal{X}, \mathcal{X}') = p_1^{r_1}\cdots p_s^{r_s}$. Then $\V_L(E_{ij})\setminus \V_L(N_{ij})$ can be divided into
 	$$\V_L(E_{ij})\setminus \V_L(N_{ij})= \V_L(E_{ij})\setminus \V_L(N_{ij}\times (p_{1}\cdots p_s)) \cup \bigcup_{j=1}^{s}\V_L(E_{ij},p_j)\setminus \V_L(N_{ij}\times (p_{1}\cdots p_{j-1})).$$
 	Let $A_{ij}:=A_{ij} \cup \{(E_{ij},N_{ij} \times (p_{1}\cdots p_s),t)\}$ and $S:=S\backslash \{t\}$. For each $\V_L(E_{ij},p_j)\setminus \V_L(N_{ij}\times (p_{1}\cdots p_{j-1}))$, if $\V_L(E_{ij},p_j)\backslash \V_L(N_{ij}\times (p_{1}\cdots p_{j-1})) \neq \emptyset$, do the step (3.b). This process may increase the number of branches, but for each branch the polynomials are simpler.
 	\item The choice of separating element may affect the performance of the algorithm significantly. How to choose a proper separating element is a complicated issue, so we do not discuss here. In our implementation, we set a variate called {\textbf{Sep$\underline{~}$ele$\underline{~}$set}}. The user can choose the separating elements manually or automatically. If the user set {\textbf{Sep$\underline{~}$ele$\underline{~}$set}} be an empty set, then the separating elements will be chosen orderly in $\{x_{1},x_{1}+x_{2}+\cdots+x_n,x_{1}+2x_{2}+\cdots+2^{n-1}x_n,x_{1}+3x_{2}+\cdots+3^{n-1}x_n,\cdots\}$. If {\textbf{Sep$\underline{~}$ele$\underline{~}$set}} is not empty, we will test the element in {\textbf{Sep$\underline{~}$ele$\underline{~}$set}} preferentially. In the examples, we set Sep$\underline{~}$ele$\underline{~}$set $= \{x_1, x_3, x_2+x_3, x_1-x_3,  x_1+x_3, x_2+x_3, x_2-x_3\}$ in $F_7$, {Sep$\underline{~}$ele$\underline{~}$set} = $\{x_4, x_1+x_4, x_3+x_4, x_1-x_4, x_3-x_4\}$ in $F_8$, Sep$\underline{~}$ele$\underline{~}$set $ = \{x_1,  x_1+2x_2, x_2+2x_3, x_1+2x_3\}$ and Sep$\underline{~}$ele$\underline{~}$set $=\{x_{2}\}$ in $S_{15}$. In other examples, the separating elements are chosen automatically.
 	\item In implementation, we also output the conditions when the system has no solution and the conditions when the system is positive dimensional. In this paper, the zero-dimensional branch does not include the branch which corresponds to a system with no solutions. If the parametric polynomial system has no zero-dimensional branches, we will output: There are no zero-dimensional branches. (For example, $F_1,F_3,S_3$ and $S_{14}$).
 	\item  In formula \ref{eq3} and \ref{eq4}, we store the values of   $\text{Tr}(M_{t^i}^I)$, $\text{Tr}(M_{x_kt^i}^I)$ after computation since they may be used in other branches.
 	\item In practice, in Step 2.(b), we compute the minimal comprehensive \gr system of the module generated by the row of $Q$ instead of introducing new variables $z_i$. It is more efficient to compute over module than introducing new variables.
 \end{itemize}

\section{Concluding remarks}

 The rational univariate representation of zero-dimensional ideals with parameters has been considered in the paper. Because the number of zeros for zero-dimensional ideals with parameters under parametric specializations is different, the choosing of separating elements which is the premise and the key to computing the rational univariate representation is quite difficult. By means of comprehensive \gr systems to divide the parameter space, we make the ideal under each branch have the same number of zeros. Moreover, we extended the subresultant theorem to parametric cases, and based on it we proposed two ideas to choose and check the separating element corresponding to each branch, which further divides the parameter space. As a result, we obtained a finite set of which each branch shares the same expression of rational univariate representation and proposed two algorithms for computing rational univariate representations of parametric zero-dimensional ideals.

\section*{Acknowledgments}

 This research was supported by the National Natural Science Foundation of China under Grant Nos. 12171469 and 12201210, and the National Key Research and Development Project 2020YFA0712300.

\bibliographystyle{elsarticle-harv}

\bibliography{PRUR_ref}
\section*{Appendix: Examples.}
	\noindent 1. Examples in \cite{Nabeshima07}:

	$$\begin{aligned}
		F_{1}&=\left\{u_{1} x_{1}^{4} x_2+x_1 x_2^{2}+u_2 x_1, x_1^{3}+2 x_1 x_2, u_2 x_1^{2}+x_1^{2} x_2\right\}. \\
		F_{2}&=\left\{u_{1} x_1^{2} x_2^{3}+u_2 x_2+x_2, x_1^{2} x_2^{2}+x_1 x_2+2, u_{1} x_1^{2}+u_2 x_2+2\right\}. \\
		F_{3}&=\left\{u_{1} x_1^{4}+u_3 x_1^{2}+u_2, u_2 x_1^{3}+x_1^{2}+2, u_3 x_1^{2}+u_4 x_1\right\}. \\
		F_{4}&=\left\{u_{1} x_1^{3} x_2+u_3 x_1 x_2^{2}, x_1^{4} x_2+3 u_4 x_2, u_3 x_1^{2}+u_2 x_1 x_2, x_1^{2} x_2^{2}+u_{1} x_1^{2}, x_1^{5}+x_2^{5}\right\}\\
		F_{5}&=\left\{u_{1} x_1^{2} x_2+u_2 x_1+x_2^{3}, u_{1} x_1^{2} x_2+u_2 x_1 x_2, x_2^{2}+u_2 x_1^{2} x_2+u_3 x_1 x_2\right\}. \\
		F_{6}&=\left\{x_1^{4}+u_{1} x_1^{3}+u_2 x_1^{2}+u_3 x_1+u_4, 4 x_1^{3}+3 u_{1} x_1^{2}+2 u_2 x_1+u_3\right\}. \\
		F_{7}&=\left\{x_1^{3}-u_{1}, x_2^{4}-u_2, x_1+x_2-u_{1} x_3\right\}. \\
		F_{8}&=\left\{u_{1} x_1^{2}+u_2 x_3, u_3 x_4^{2}+x_2,(x_1-x_2)^{2}+(x_3-x_4)^{2}, 2 u_4 x_1 x_4-2 u_2 x_3\right\}.
	\end{aligned}$$

	\vspace{0.2cm}
	
	\noindent 2. Examples in \cite{Montes06}:
	$$\begin{aligned}
		S_1&=\left\{u_1 (x_1+x_2), u_2 (x_1+x_2), x_1^2+u_1 x_1\right\}.\\
		S_2&=\left\{x_1^2, x_1 x_2,x_1 x_3^2, u_1 x_1+x_2, x_2 x_3-x_3^2,u_1 x_2, x_3^3,u_1 x_3^2, u_1^2\right\}.\\
		S_3&=\left\{x_1^3-u_1 x_1 x_2, x_1^2 x_2-2 x_2^2+x_1\right\}.\\
		S_4&=\{u_1 x_1+x_2-1, u_2 x_1+x_2-2, 2 x_1+u_1 x_2, u_2 x_1+u_1 x_2+1\}.\\
		S_5&=\{x_4-(u_4-u_2), x_1+x_2+x_3+x_4-(u_1+u_3+u_4), x_1 x_3 x_4-u_1 u_3 u_4, \\
		&\quad \; \;x_1 x_3+x_1 x_4+x_2 x_3+x_3 x_4-(u_1 u_4+u_1 u_3+u_3 u_4)\}.\\
		S_6&=\left\{u_2 x_1 x_2+u_1 x_1^2+x_1,u_1 x_2^2+x_1^2\right\}.\\
		S_7&=\left\{x_2^2-x_1 x_2 x_3+x_1^2+x_3-1,x_1 x_2+x_3^2-1,x_1^2+x_2^2+x_3^2-u_1^2\right\}.\\
		S_8 &= \left\{u_1-u_2+(u_1x_1 x_2 -u_2x_1^{2} x_2-3u_1)^{3}+(x_1 x_2 u_2-3 x_1 u_2-5 u_2)^{4},\right.\\
		& \quad \; \; \left. u_1x_1 x_2 -u_2x_1^{2} x_2 -3 u_1,u_2x_1 x_2 -3 u_2x_1 -5 u_2\right\}.\\
		S_9 &= \left\{x_1+u_3 x_2+u_2 x_3+u_1, u_3 x_1+x_2+u_1 x_3+u_2, u_2 x_1+u_1 x_2+x_3+u_3\right\}.\\
		S_{10} &= \left\{x_1^2+x_3^2-1, x_2^2+x_4^2-1, u_3 (x_1 x_2-x_3 x_4)-x_3+u_1,u_3 (x_1 x_4+x_3 x_2)+x_1-u_2\right\}.\\
		S_{12} &=\left\{u_1-u_3 x_2-u_4 x_1,u_2-u_3 x_4-u_4 x_3  ,x_1^2+x_3^2-1,x_2^2+x_4^2-1\right\}.\\
		S_{13} &=\left\{u_1 x_1^2 x_2+u_1 3 u_2^2,u_1 (u_2-u_3) x_1 x_2+u_1 u_2 x_1+5 u_3\right\}.\\
		S_{14} &=\left\{ x_4^3-u_1 x_1 x_4^2-x_1 x_2^2-x_1 x_3^2, x_4^3-u_1 x_2 x_4^2-x_2 x_1^2-x_2 x_4^2, x_4^3-u_1 x_3 x_4^2-x_3 x_1^2-x_3 x_2^2\right\}.\\
		S_{15} &=\left\{u_1+u_4 x_1, u_2-u_4 x_4, u_5 x_5+u_6 x_6-u_4,u_5 x_2+u_6 x_3-u_3,x_1^2+x_4^2-1,x_2^2+x_5^2-1,x_3^2+x_6^2-1\right\}.\\
		S_{16} &= \text{See Section 5.2 in \cite{Montes06}}.\\
	\end{aligned}$$
	
	\vspace{0.2cm}
	
	\noindent 3. Other examples: $C_1$ is from \text{\cite{Jerome07}}, $R_1$ is from \cite{Ayad12}, $E_1,E_2,E_3$ are from \cite{Robbiano13}.
	$$\begin{aligned}
		C_1&=\left\{u_1-3 x_1^2-4 x_1 x_2-2 x_2^2-2 x_1 x_3-2 x_2 x_3-x_3^2,u_2 - 2 x_1^2-3 x_1 x_2-x_1 x_3-x_2^2-x_2 x_3,u_3-x_1^2-x_1 x_2-x_1 x_3\right\}.\\
		R_1 &= \left\{x_1+x_2-u_1,x_3+x_4-u_2, x_1^2+x_3^2-1,x_2^2+x_4^2-1\right\}.\\
		E_1&= \left\{x_1^2+u_1 x_2^2-1, x_2^2+u_2 x_2, -2 u_1 u_2 x_2+4 x_1 x_2\right\}.\\
		E_2&=\left\{u_1 x_1 x_2+u_2,u_3 x_1^2+u_4 x_2^2+u_5,-2 u_1 u_3 x_1^2+2 u_1 u_4 x_2^2\right\}.\\
		E_3&= \left\{x_1 x_2+u_1 x_1+1, x_1^2+x_2^2+u_2,-2 x_1^2+2 x_2^2+2 u_1 x_2\right\}.
	\end{aligned}$$
\end{document}